\documentclass[twoside,leqno]{article}
\usepackage[letterpaper]{geometry}
\usepackage{ltexpprt}
\usepackage{hyperref}

\usepackage[utf8]{inputenc}
\usepackage{amsmath}
\usepackage[USenglish]{babel}
\usepackage{amssymb}
\usepackage{mathtools} % for \DeclarePairedDelimiter
\usepackage[textsize=footnotesize,textwidth=.8in,disable]{todonotes}
\usepackage[ruled,noend, algo2e]{algorithm2e}
\usepackage[noend]{algorithmic} 

\usepackage[shortlabels]{enumitem}
\usepackage{comment}
\usepackage{tabularx}
\usepackage{makecell}

\newtheorem{observation}{Observation}
\newtheorem{definition}{Definition}

\newtheorem{invariant}{Invariant}
\newtheorem{question}{Question}
\newcommand{\eps}{\varepsilon}
\newcommand{\mx}{\max}
\newcommand{\mout}{\Delta(\overrightarrow{G})}
\newcommand{\polylog}{\operatorname{polylog}}
\newcommand{\poly}{\operatorname{poly}}
\newcommand{\BigOh}{\mathcal{O}}

\DeclareMathOperator*{\argmin}{arg\,min}
\DeclarePairedDelimiter{\paren}{(}{)}

\DeclarePairedDelimiter{\ceil}{\lceil}{\rceil}
\DeclarePairedDelimiter{\floor}{\lfloor}{\rfloor}
\DeclarePairedDelimiter{\abs}{\lvert}{\rvert}
\DeclarePairedDelimiter{\set}{\lbrace}{\rbrace}

\newcommand{\suchthat}{\mathrel{}\mathclose{}\ifnum\currentgrouptype=16\middle\fi\vert\mathopen{}\mathrel{}}

\newcommand{\orc}{\includegraphics[height=\fontcharht\font`A]{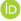}}

\newcommand\relatedversion{}
\renewcommand\relatedversion{\thanks{The full version}}
\title{\Large Adaptive Out-Orientations with Applications\relatedversion}
\author{ Chandra Chekuri \href{https://orcid.org/0000-0003-3035-1699}{\orc} \thanks{Dept. of Computer Science, University of Illinois, Urbana, IL 61801. {\tt chekuri@illinois.edu}}
\and Aleksander Bjørn Christiansen\thanks{Technical University of Denmark, Kongens Lyngby, Denmark {\tt abgch@dtu.dk, vanderhoog@gmail.com, erot@dtu.dk}}
\and Jacob Holm \href{https://orcid.org/0000-0001-6997-9251}{\orc}\thanks{University of Copenhagen, Denmark {\tt jaho@di.ku.dk}}
\and Ivor van der Hoog \href{https://orcid.org/0009-0006-2624-0231}{\orc} $^{\ddagger}$
\and Kent Quanrud\thanks{Dept. of Computer Science, Purdue University, West Lafayette, IN. {\tt krq@purdue.edu}}
\and Eva Rotenberg \href{https://orcid.org/0000-0001-5853-7909}{\orc} $^{\ddagger}$
\and Chris Schwiegelshohn\thanks{Aarhus University, Aarhus, Denmark. {\tt cschwiegelshohn@gmail.com}}
}

\date{October 2023}

\begin{document}

\maketitle

%\fancyfoot[R]{}
%\scriptsize{Copyright \textcopyright\ 2024 by SIAM\\ Unauthorized reproduction of this article is prohibited}}

\begin{abstract}
We give improved algorithms for maintaining edge-orientations of a fully-dynamic graph, such that the maximum out-degree is bounded. On one hand, we show how to orient the edges such that maximum out-degree is proportional to the arboricity $\alpha$ of the graph, in, either, an amortised update time of $\BigOh(\log^2 n \log \alpha)$, or a worst-case update time of $\BigOh(\log^3 n \log \alpha)$. 
On the other hand, 
motivated by applications including dynamic maximal matching, 
we obtain a different trade-off. Namely, the improved update time of either $\BigOh(\log n \log \alpha)$, amortised, or $\BigOh(\log ^2 n \log \alpha)$, worst-case, for the problem of maintaining an edge-orientation with at most $\BigOh(\alpha + \log n)$ out-edges per vertex.
Finally, all of our algorithms naturally limit the recourse to be polylogarithmic in $n$ and $\alpha$. 
Our algorithms adapt to the current arboricity of the graph, and yield improvements over previous work: 

Firstly, we obtain deterministic algorithms for maintaining a $(1+\varepsilon)$ approximation of the maximum subgraph density, $\rho$, of the dynamic graph. Our algorithms have update times of  $\BigOh(\varepsilon^{-6}\log^3 n \log \rho)$ worst-case, and $\BigOh(\varepsilon^{-4}\log^2 n \log \rho)$ amortised, respectively.  
We may output a subgraph $H$ of the input graph where its density is a $(1+\varepsilon)$ approximation of the maximum subgraph density in time linear in the size of the subgraph. 
These algorithms have improved update time compared to the $\BigOh(\varepsilon^{-6}\log ^4 n)$ algorithm by Sawlani and Wang from STOC 2020.

Secondly, we obtain an $\BigOh(\varepsilon^{-6}\log^3 n \log \alpha)$  worst-case update time algorithm for maintaining a $(1~+~\varepsilon)\textnormal{OPT} + 2$ approximation of the optimal out-orientation of a graph with adaptive arboricity $\alpha$, improving the $\BigOh(\varepsilon^{-6}\alpha^2 \log^3 n)$ algorithm by Christiansen and Rotenberg from ICALP 2022. 
This yields the first worst-case polylogarithmic dynamic algorithm for decomposing into $\BigOh(\alpha)$ forests.

Thirdly, we obtain arboricity-adaptive fully-dynamic deterministic algorithms for a variety of problems including maximal matching, $\Delta+1$ colouring, and matrix vector multiplication. All update times are worst-case $\BigOh(\alpha+\log^2n \log \alpha)$, where $\alpha$ is the current arboricity of the graph. 
For the maximal matching problem, the state-of-the-art deterministic algorithms by Kopelowitz, Krauthgamer, Porat, and Solomon from ICALP 2014 runs in time $\BigOh(\alpha^2 + \log^2 n)$, and by Neiman and Solomon from STOC 2013 runs in time $\BigOh(\sqrt{m})$. We give improved running times whenever the arboricity $\alpha \in \omega( \log n\sqrt{\log\log n})$.
\end{abstract}

\paragraph{Acknowledgements.}
This research was supported by Independent Research Fund Denmark grant 2020-2023 (9131-00044B) ``Dynamic Network Analysis'' and the VILLUM Foundation grant (VIL37507) ``Efficient Recomputations for Changeful Problems''. 
This project has additionally received funding from the European Union's Horizon 2020 research and innovation programme under the Marie Sk\l{}odowska-Curie grant agreement No 899987. 
Chris Schwiegelshohn is partially %partially?
supported by an Independent Research Fund Denmark (DFF) Sapere Aude Research Leader grant
No 1051-00106B.
Chandra Chekuri is supported by NSF grant CCF-1910149.  Kent Quanrud is supported in part by NSF grant CCF-2129816.

\newpage

\section{Introduction}
\label{sec:intro}
In dynamic graphs, one wishes to update a data structure over a graph $G(V,E)$ (or an answer to a specified graph problem) as the graph undergoes local updates such as edge insertions and deletions. 
One of the fundamental problems is to maintain an orientation of the edges such that the maximum out-degree over all vertices is minimised.
While the problem is interesting in its own right, bounded out-degree orientations have a number of applications. First, the problem is closely related to the task of finding the densest subgraph; indeed if the edges can be fractionally oriented, the optimal maximal fractional out-degree is equal to the density $\rho := \frac{|E\cap (S\times S)|}{|S|}$ of the densest subgraph $S\subseteq V$.
Secondly, bounded out-degree orientations appear frequently as subroutines for other problems. In particular, there exist a large body of work parameterising the update time of dynamic algorithms for many fundamental problems such as shortest paths \cite{FrigioniMN03,KowalikK06}, vertex cover \cite{PelegS16,Solomon18}, graph colouring \cite{henzinger2020explicit,SolomonW20}, independent set \cite{OnakSSW20}, and, most prominently, maximum matching \cite{BernsteinS15,BernsteinS16,GSSU22,HeTZ14,NeimanS16,PelegS16,Solomon18} in terms of the arboricity $\alpha:=\max_{S\subseteq V,|S|\geq2} \ceil[\big]{\frac{|E\cap (S\times S)|}{|S|-1}}$.

In light of their widespread applicability, maintaining an edge orientation minimising the maximum outdegree is extremely well motivated. In particular, we are interested in algorithms with worst-case deterministic update times, as these can be immediately used as black-box subroutines. In a recent breakthrough result, \cite{sawlani2020near} showed that it is possible to maintain an estimate for the smallest maximum outdegree in $\polylog(n)$ worst case deterministic time by maintaining an estimate for the density of the densest subgraph. Nevertheless, all known results for maintaining an orientation require at least update time $\Omega(\rho) = \Omega(\alpha)$ worst case update time, regardless of whether the algorithm is randomised or not \cite{KopelowitzKPS13}. For dense graphs, this bound may be arbitrarily close to $n$. Thus, it raises the following question:

\begin{question}
    Is it possible to maintain an (approximate) minimum out-degree orientation in sublinear deterministic worst case update time?
\end{question}

\subsection{Our Contribution}

In this paper, we answer the aforementioned question in the affirmative. Specifically, we provide a framework for maintaining approximate out-orientations with various trade-offs between the quality of the out-degree orientation and update time. 
For the problem of maintaining an out-orientation we obtain:
\begin{enumerate}[noitemsep]
    \item An orientation with maximum out-degree $\BigOh(\alpha)$ in update time $\BigOh(\log^3 n \log \alpha)$.
    \item An orientation with maximum out-degree $(1+\varepsilon)\alpha +2$ in update time $\BigOh(\varepsilon^{-6}\log^3 n \log \alpha)$.
    \item An orientation with maximum out-degree $\BigOh(\alpha +\log n)$ in update time $\BigOh(\log ^2 n \log \alpha)$.
\end{enumerate}

\noindent The above running times are deterministic and worst-case. 
Contrary to the previous state-of-the-art result by Sawlani and Wang~\cite{sawlani2020near}, the recourse of our algorithm, i.e. the number of re-orientations of edges, is polylogarithmic in $n$ (specifically, a $\log n$ factor lower than the running time). 
\\

\noindent
When allowing \emph{amortisation}, we get even better bounds:
\begin{enumerate}[noitemsep]
\setcounter{enumi}{3}
    \item An orientation with maximum out-degree $\BigOh(\alpha)$ in amortised update time $\BigOh(\log^2 n \log \alpha)$.
    \item An orientation with maximum out-degree $\BigOh(\alpha +\log n)$ in amortised update time $\BigOh(\log n \log \alpha)$.
\end{enumerate}

Table~\ref{tab:resultssmall} gives an overview of our results, and their implications when applied to a selection of algorithmic problems.
The latter we briefly discuss in the following.

\subparagraph{Densest Subgraph}
Using the duality between out-degree orientations and maximum subgraph density, we obtain a $(1+\varepsilon)$ approximate estimate for maximum subgraph density $\rho$ in worst-case update time of  $\BigOh(\varepsilon^{-6}\log^3 n \log \rho)$. 
Additionally, we may output a subgraph $H$ with a density greater than $\rho / (1 + \eps)$ in time linear in the size of $H$ (Lemma~\ref{lemma:SDE}). 
This recovers (and moderately improves) the recent worst-case algorithm by \cite{sawlani2020near} that has an update time of $\BigOh(\varepsilon^{-6}\log^4 n)$.
When allowing amortised analysis, we improve the running time to  $\BigOh(\varepsilon^{-4}\log^2 n\log\rho)$ amortised.

\subparagraph{Arboricity Decomposition}
An arboricity decomposition partitions the edge set into a minimum number of forests. The best dynamic algorithms for maintaining an $\BigOh(\alpha)$ 
arboricity decomposition has an amortised deterministic update time of $\BigOh(\log^2 n)$ due to \cite{henzinger2020explicit} and an $\BigOh(\sqrt{m}\log n)$ worst case deterministic update time due to \cite{Banerjee}. 
Distinguishing between arboricities $1$ and $2$ requires $\Omega(\log n)$ time \cite{Patrascu10,Banerjee}.
We substantially improve the worst case update time to $\BigOh(\log^3 n \log \alpha)$.

\subparagraph{Dynamic Matrix Vector Multiplication}
In the Dynamic Matrix Vector Multiplication problem, we are given an $n\times n$ matrix $A$ and an $n$-vector $x$. Our goal is to quickly maintain $y=Ax$ in the sense that we can quickly query every entry $y_i = (Ax)_i$, subject to additive updates to the entries of $x$ and $A$. 
Interpreting $A$ as the adjacency matrix of a graph with arboricity $\alpha$, \cite{KopelowitzKPS13} presented an algorithm supporting updates to $A$ in time $\BigOh(\alpha^2+\log^2 n)$ and updates to $x$ in time $\BigOh(\alpha + \log n)$. We may update $A$ in time $\BigOh(\alpha + \log ^2 n \log \alpha)$, improving when $\alpha\in\omega(\log n\sqrt{\log\log n})$.

\subparagraph{Maximal Matching}
A matching is a set of vertex-disjoint edges. A matching $M$ is maximal if no edge of the graph can be added to it without violating the matching property. 
More so than perhaps any other problem, there exists a large gap between the performance of the state of the art deterministic algorithms vs the state of the art randomised algorithms. Using randomisation, one can achieve a $\BigOh(1)$ amortised \cite{Solomon16} and a $\polylog(n)$ worst case update time \cite{BernsteinFH21}. Deterministic algorithms so far have only achieved a $\BigOh(\sqrt{m})$ update time for arbitrary graphs \cite{NeimanS16}, or $\BigOh(\alpha^2 + \log^2 n)$ update time where $\BigOh(\alpha)$ is the current arboricity of the graph \cite{KopelowitzKPS13}. 
Because our result explicitly maintains an (approximately) optimal orientation, we improve on known deterministic algorithms whenever $\alpha \in \omega(\log n \cdot \sqrt{\log\log n})$ by achieving an update time of $\BigOh(\alpha + \log^2 n\log \alpha)$.

\subparagraph{$\Delta+1$ Colouring}
A fundamental question in many models of computation is how to efficiently compute a $\Delta+1$ colouring where $\Delta$ is the maximum degree of the graph. 
We present a deterministic algorithm that maintains a $\Delta+1$ colouring in $\BigOh(\alpha+\log^2 n\log \alpha)$ worst case update time. To the best of our knowledge, this is the first such algorithm that beats the trivial $\BigOh(\Delta)$ update time for uniformly sparse graphs. All other results~\cite{BhattacharyaGKL22,HenzingerP22,BhattacharyaCHN18,SolomonW20} (discussed in more detail in the appendix) require randomisation, amortisation, or do not yield a $\Delta+1$ colouring.

\begin{table}[ht]
\vspace{1em}
    \centering
    \begin{tabularx}{1.1\textwidth}{c|c|c|c|l}
      \makecell{\textbf{Dynamic} \\
     \textbf{Problem} 
     } & \textbf{Guarantee} & \makecell{\textbf{ Worst case }\\ 
\textcolor{green!60!black}{\textbf{Amortised}}} & \textbf{Thm.} & \textbf{State-of-the-art comparison} \\
         \hline
         \hline
     %
     %   CONSTANT FACTOR
     %
    \makecell{out-orient \\ /density}
     & 
     %\makecell{  
     $\BigOh(\alpha)$ 
     %\\  $\BigOh(\rho)$}    
     & \makecell{$\BigOh( \log^3 n \log \alpha)$ \\
    \textcolor{green!60!black}{$\BigOh(\log^2 n \log \alpha )$}} & \makecell{Thm.~\ref{thm:rank0}  \\
    \textcolor{green!60!black}{Thm.~\ref{thm:amor0}}} & 
     \makecell[l]{
     \textcolor{green!60!black}{ $\BigOh(\alpha_{\mx}) \textnormal{ in } \BigOh(\alpha_{\mx} + \log n) \textnormal{~\cite{Brodal99dynamicrepresentations}}$} \\
      \textcolor{green!60!black}{$\BigOh(\alpha) \textnormal{ in } \BigOh(\log^2 n \log \alpha)  \textnormal{~\cite{henzinger2020explicit}}$} \\
      \textcolor{blue!60!black}{$\BigOh(\alpha) \textnormal{ in } \BigOh(\log^4 n) \textnormal{~\cite{sawlani2020near}}$}
     }\\
     \hline
     % 
     % ADDITIVE LOG
     %
     \makecell{out-orient \\ /density} & 
      %\makecell{  
      $\BigOh(\alpha + \log n)$
      
      & \makecell{$\BigOh(\log ^2 n \log \alpha)$ \\ \textcolor{green!60!black}{$\BigOh(\log n \log \alpha)$}} & \makecell{ Thm.~\ref{thm:rank1} \\ \textcolor{green!60!black}{Thm.~\ref{thm:amor1}}} &      
     \makecell[l]{
       $\BigOh(\alpha_{\mx} + \log n) \textnormal{ in } \BigOh(\log n)  \textnormal{~\cite{berglinetal:LIPIcs:2017:8263}}$ \\
      $\BigOh(\alpha + \log n) \textnormal{ in } \BigOh(  \alpha^2 + \log^2 n ) \textnormal{~\cite{KopelowitzKPS13}}$\\
      \textcolor{green!60!black}{$\BigOh(\alpha) \textnormal{ in } \BigOh(\log^4 n) \textnormal{~\cite{sawlani2020near}}$}
} \\
     \hline
     %
     % EPSILON DENSITY
     %
     density & $ (1 + \eps) \rho $ & \makecell{$\BigOh\paren*{\varepsilon^{-6}\log^3 n \log \rho}$ \\
     \textcolor{green!60!black}{$\BigOh\paren*{\varepsilon^{-4}\log^2 n \log \rho}$ }} & Cor.~\ref{cor:eps_subgraphdensity} & 
     \makecell[l]{
      $(1 + \eps) \rho \textnormal{ in } \BigOh(\eps^{-6} \log^4 n  ) \textnormal{~\cite{sawlani2020near}} $
     } \\
     \hline
    %
     % 2 EPSILON OUT 
     %
    out-orient & $(2 + \eps) \alpha$ & \makecell{$\BigOh\paren*{\varepsilon^{-6}\log^3 n \log \alpha}$ \\  \textcolor{green!60!black}{$\BigOh\paren*{\varepsilon^{-4}\log^2 n \log \alpha}$}} & Obs.~\ref{obs:rounding} & 
       \textcolor{blue!60!black}{ $(2 + \eps) \alpha \textnormal{ in } \BigOh(\eps^{-6} \log^4 n  )  \textnormal{~\cite{sawlani2020near}}$}
     \\
     %%%%%
     % 1 EPS OUT
     %%%%
     \hline
     out-orient & $(1 + \eps)\alpha + 2$ & $\BigOh( \eps^{-6} \log^3 n \log \alpha)$ & Thm.~\ref{thm:epsapprox} & 
     $
     \BigOh( \eps^{-6} \alpha^2 \log^3 n) \textnormal{~\cite{christiansenICALP}}
     $ 
     \\ 
     %
     % MAX MATCHING
     %
     \hline
      \makecell{matching}  & maximal & $\BigOh(\alpha + \log ^2 n \log \alpha)$ & Cor.~\ref{cor:match} & 
     \makecell[l]{%$%
%    \BigOh(\alpha_{\mx} + \log n) \textnormal{~ \cite{berglinetal:LIPIcs:2017:8263}}$ \\
 %      $\BigOh(\alpha_{\mx}+\sqrt{\alpha_{\mx}\log n} ) \sim\textnormal{amor.~\cite{HeTZ14}}$\\
       $\BigOh(\sqrt{m} )  \textnormal{~\cite{NeimanS16}}$ \\
     $\BigOh(\alpha^2 + \log^2 n) \textnormal{~\cite{KopelowitzKPS13}}%
     $}
     \\
     \hline
    \makecell{
     coloring } 
     &  $\Delta+1$ colors & $\BigOh(\alpha+\log^2 n\log \alpha)$ & Cor.~\ref{cor:col} &
     $\BigOh(\Delta)$ (folklore)
     \\
     %
     % FOREST DECOMP
     %
     \hline
     \makecell{ arboricity  \\ decomp} 
     & $\BigOh(\alpha)$ & $\BigOh(\log^3 n \log \alpha)$ & Cor.~\ref{cor:decomp} &
     \makecell[l]{
     \textcolor{green!60!black}{$\BigOh(\alpha) \textnormal{ in } \BigOh(\log^2 n)$
     \textnormal{~}\cite{henzinger2020explicit}} \\ 
     $\BigOh(\sqrt{m}\log n)$ \cite{Banerjee} 
     } \\
       \hline
     \makecell{ maintain \\
     $A \cdot \overrightarrow{x}$} & -- & \makecell{$\BigOh(\log ^2 n \log \alpha)$ 
     upd. $A$ 
     \\
     $\BigOh(\alpha + \log n)$  \textnormal{ upd. }  $\overrightarrow{x}$ \\
     $\BigOh(\alpha + \log n)$  \textnormal{ query } 
     } & Cor.~\ref{cor:dynamicmatrix} &
  \makecell[l]{
     $\BigOh(\alpha^2 + \log^2 n) \textnormal{ updating } A$ \\
     $\BigOh(\alpha + \log n) \textnormal{ query+upd } \overrightarrow{x} 
     $
     \cite{KopelowitzKPS13}
    }
    \end{tabularx}
    \caption{
    An abbreviated overview of our results, where we compare to state-of-the-art deterministic results. Our results are adaptive to the arboricity $\alpha$ and the maximal subgraph density $\rho$, and explicitly maintain the out-orientation with low recourse. For out-orientations, note that \cite{sawlani2020near} does not explicitly re-orient edges between updates, but allow for the orientations of edges to be computed upon query time (hence, we marked their result blue).
    \label{tab:resultssmall}
    }
\end{table}

%\subsection{Related Work}
\subparagraph{Related work on dynamic orientations}
Dynamic out-orientations have been widely studied~\cite{Brodal99dynamicrepresentations,HeTZ14,10.1016/j.ipl.2006.12.006,christiansenICALP,christiansenMFCS,berglinetal:LIPIcs:2017:8263} since they were introduced by Brodal and Fagerberg~\cite{Brodal99dynamicrepresentations}, for maintaining an $\BigOh(\alpha_{\mx})$ out-orientation\footnote{Here $\alpha_{\max}$ denotes the maximum arboricity seen over the whole sequence of operations.} in $\BigOh(\alpha_{\max}+\log n)$ amortised time. 
Brodal and Berglin~\cite{berglinetal:LIPIcs:2017:8263} improve the time guarantee to worst-case $\BigOh(\alpha_{\max}+\log n)$  time, albeit maintaining an $\BigOh(\alpha_{\max}+\log n)$ out-orientation.
The best \emph{adaptive} algorithms, adapting to a changing arboricity, are by Henzinger, Neumann, and Wiese~\cite{henzinger2020explicit} achieving an out-degree of $\BigOh(\alpha)$ and an amortised update time $\BigOh(\log^2 n)$, and by Kopelowitz, Krauthgamer, Porat, and Solomon~\cite{KopelowitzKPS13}, maintaining an $\BigOh(\alpha+\log n)$ out-orientation with a worst-case update time of $\BigOh(\alpha^2 + \log^2 n)$. 
Christiansen and Rotenberg~\cite{christiansenMFCS,christiansenICALP} lowered the maximum out-degree to $(1+\varepsilon)\alpha+2$ incurring a worse update time of $\BigOh(\varepsilon^{-6}\alpha^2\log^3 n)$.

\section{Notation and overview of techniques}\label{sec:tech}

Let $G =(V, E)$ be a graph with $n$ vertices and $m$ edges. 
For any subgraph $H$ of $G$, we denote by $V(H)$ and $E(H)$ the corresponding vertex and edge set.
The \emph{density} of subgraph $H$ is $\rho(H) := \frac{|E[G]|}{|V(G)|}$. The \emph{maximum subgraph density} of $G$ is then the maximum over all $H$ of $\rho(H)$. A closely related measure of uniform sparsity is the \emph{arboricity} of a graph, defined as:
\[
\alpha := \max \limits_{H \subseteq G , |V(H)| \geq 2} \ceil*{ \frac{|E(H)|}{|V(H)|-1} } 
\]

A \emph{fractional} orientation $\overrightarrow{G}$ of a graph $G$ assigns for every edge $(u, v)$ a weight $d(u \to v)$ and $d(v \to u)$ such that $d(u \to v) + d(u \to v) = 1$. 
The out-degree of a vertex $u$ is subsequently defined as $d^+(u) = \sum_{v \in V} d(u \to v)$. 
The maximum out-degree of $\overrightarrow{G}$ is $\Delta(\overrightarrow{G}) := \max_{v \in V} d^+(v)$. 
Picard \& Queyranne~\cite{PicardQueyranne82} show that $\lceil \rho \rceil = \min_{\overrightarrow{G}} \Delta(\overrightarrow{G})$, and so it follows that $\rho \leq \Delta(\overrightarrow{G})$.

An \emph{orientation} is a fractional orientation where $d(u \to v)$ is either $0$ or $1$.
For brevity, we say that an orientation includes $\overrightarrow{uv}$ when  $d(u \to v) = 1$.

For any vertex $u \in V$, we subsequently denote by $N^+(u)$ (resp. $N^-(u)$) all vertices $w$ with $d(u \to w) = 1$  (resp. $d(w \to u) = 1$). 
In an orientation, $d^+(u)$ is the number of edges directed from $u$ (the out-degree). 
Whenever $G$ is not simple, $d^+(u)$ can be larger than $|N^+(u)|$.

For any integer $b \geq 1$, we denote by $G^b$ the graph $G$ where every edge is duplicated $b$ times.
Throughout this paper we maintain for a suitable choice of $b$, an orientation over $G^b$. 
Note that any orientation in $G^b$ induces a fractional orientation on $G$. 
We may convert any such fractional orientation in $G$ to an orientation on $G$ by `rounding' every edge (i.e., $d(u \to v) = 1$ if $d(u \to v) > d(v \to u)$, breaking ties arbitrarily). 
Observe that if the maximum out-degree of an orientation in $G^b$ is some value $\Delta$, then the maximum out-degree of the rounded orientation in $G$ is at most $\Delta/\ceil{\frac{b}{2}}\leq 2 \Delta / b$. 
An important theoretical insight for this work are the following linear programs.
These dual programs respectively maximise the subgraph density, or minimise the largest fractional out-degree of an edge orientation of $G$:

\vskip\baselineskip
\noindent\begin{minipage}[t]{.5\textwidth}
\noindent\textbf{Densest Subgraph (DS)}
\begin{align*}
\text{maximise } &\sum_{ \overline{uv} \in E}  y_{u, v}
\qquad\text{s.t.}\\
x_u,x_v &\geq y_{u,v}\qquad\forall u,v\in V, \overline{uv}\in E\\
\sum_{v\in V}x_v&\leq 1\\
x,y&\geq 0
\end{align*}    
\end{minipage}%
\begin{minipage}[t]{.5\textwidth}
\noindent\textbf{Fractional Orientation (FO)}
\begin{align*}
\text{minimise } &\rho \qquad\text{s.t.}\\
d(u\to v)+d(v\to u) &= 1&\forall\overline{uv}\in E\\
\rho \geq d^+(u) &= \sum_{v\in V} d(u \to v) &\forall u\in V \\
\rho,d(u\to v),d(v\to u)&\geq 0
\end{align*}
\end{minipage}
\vskip\baselineskip

% \begin{align*}
%    \textbf{Densest Subgraph (DS) } & &  & & \textnormal{~} & \hspace{3em} & \textbf{Fractional Orientation (FO)} & & \\
%  \max \sum_{ \overline{uv} \in E} d(\overline{uv}) \cdot y_{u, v} & & \text{s.t.} & &  & & \min \rho  & & \text{s.t.} \\
% x_u,x_v \geq y_{u,v}  & & \forall u,v\in V, \overline{uv} \in E & & \textnormal{~} & & d(u \to v) + d(v \to u) = 1 & &\forall \overline{uv}\in E \\
% \sum_{v\in V} x_v  \leq 1 & & & & \textnormal{~} & & \rho \geq d^+(u) = \sum_{v\in V} d(u \to v) & &\forall u\in V \\
% x,y \geq 0 & & & & \textnormal{~} & & \rho,d(u \to v), d(v \to u) \geq 0 & &
% \end{align*}

\paragraph{Duality and previous work.}
The duality between these programs allows for approximating the maximum subgraph density by computing a fractional orientation that aims to minimise $\rho$. 
Thus, in an algorithmic sense, we focus on maintaining a fractional orientation of $G$.
This is then achieved by maintaining an integral orientation in a graph with an appropriate number of edge duplicates.

These integral orientations are typically maintained using the following simple, but efficient idea: If one takes a directed path from a high-out-degree vertex to a low-out-degree vertex, then reorienting every edge along this path lowers the out-degree of the high-out-degree vertex while only increasing the out-degree of some vertex of low out-degree. 
To make this idea constructive, one needs a way to efficiently locate a suitable directed path or \emph{chain} to reorient. 

Kopelowitz et al.\cite{KopelowitzKPS13} showed how to locate such chains by maintaining a local condition, namely that $d(u \to v) > 0$ implies that $d^+(u) \leq d^+(v) + 1$. 
When the maximum out-degree is small, this local condition can be used to identify short chains. However, when the out-degree becomes large (in dense graphs) this procedure becomes slow. In particular, one can never hope to get a better bound on the chain length than $\Omega(\rho)$.
This in turn means that their update times are $\Omega(\rho) = \Omega(n)$. In fact, all of their algorithms have update times that depend on $\rho^2$. One $\rho$ stems from the chain lengths and the other from the fact that changes in degrees need to be reported to all out-neighbours in order to efficiently locate the chains. 

Sawlani and Wang~\cite{sawlani2020near} removed the latter $\rho$-factor by informing neighbours via a round robin scheme. 
They then removed the former $\rho$-factor by instead requiring that $d(u \to v) > 0$ implies that $d^+(u) \leq d^+(v) + f(\tilde{\rho})$ for some function $f$ and some very precise estimate $\tilde{\rho}$ of the \emph{current} maximum subgraph density. 
By making the local condition depend on $\tilde{\rho}$, they were able to get chains of much shorter length, namely of length $\BigOh(\eps^{-2} \log n)$. 
However, for this local condition to yield a small out-degree, one requires that $\tilde{\rho}$ very precisely estimates the current density. 
To enforce this, Sawlani and Wang~\cite{sawlani2020near} maintain $\BigOh(\frac{\log n}{\eps})$ different copies of the graph -- each with a different estimate $\tilde{\rho}$.
They maintain a pointer to the copy which currently estimates $\rho$ the best. 

While this allows Sawlani and Wang~\cite{sawlani2020near} to estimate the current maximum subgraph density very well, their approach has several drawbacks. 
First of all, their algorithm only maintains an \emph{implicit} orientation of the graph in the sense that the algorithm often switches between different copies of the graph each endowed with possibly very different orientations. 
While this does not matter in the context of density estimation, it matters in the context of using the out-orientation as an algorithmic tool. 
Firstly, any application run on such an orientation only maintains an implicit representation of the desired outcome, since one continually changes between different copies as updates arrive. 
Secondly, one has to update the applications across \emph{all} copies meaning that the guarantee on the out-degree is no better than $\BigOh(n)$ in the top-most copy -- even if the maximum subgraph density is low. 
The use of copies also  makes the algorithm significantly more complicated. 

\paragraph{Our key idea.}
We show that maintaining a multiplicative local condition, namely that $d^+(u) \leq (1 + a) \cdot d^+(v)$ for some chosen value $a < 1$, allows one to get \emph{both} short chains of length $\BigOh(\eps^{-2} \log n \log \rho)$ whilst maintaining a very precise approximation of the maximum subgraph density. 
Furthermore, this can be achieved completely \emph{explicitly} with low recourse and using only one copy of the graph. 
This allows us to apply our result to problems that benefit from having an explicit low out-degree orientation such as dynamic maximal matchings and $(\Delta + 1)$ colouring.  

Since this multiplicative local condition removes the need for scheduling updates to different copies, the algorithms also become simpler.
However, the analysis become significantly more delicate. 
Sawlani and Wang~\cite{sawlani2020near} work with an additive local condition, where the added quantity depends on a very precise estimate $\tilde{\rho}$ of the current density. 
This allows them, in many places, to essentially reduce the problem-complexity to the case where: $a)$ one has to basically only consider vertices with very large out-degree and $b)$ one can essentially assume that these out-degrees are unchanged, since one is working with a quantity depending on $\tilde{\rho}$. 
Working with our local condition, however, allows for neither simplification. 
The local condition is equally ''tight'' for every vertex, and it is very sensitive to changes in degrees at both endpoints of an edge. 
This means that one has to be very careful, when analysing the algorithms -- especially when the vertices have low degree. 

\paragraph{Multiplicative local conditions}
We consider two local conditions that we want to maintain for an integral orientation.
The first has both an additive and multiplicative term.
The second has only a multiplicative term. 
In the first case, we require that $d(u \to v) > 0$ implies that $d^+(u) \leq (1 + a) \cdot d^+(v) + c$. 
The benefit of the additive term is that for any new edge $(u, v)$, we may always orient the edge towards either $u$ or $v$ without violating this local condition between $u$ and $v$. 
The downside of this approach is that it leads to less accurate estimations of $\alpha$ and $\rho$. 
In the second case, we require that $d(u \to v) > 0$ implies that $d^+(u) \leq (1 + a) \cdot d^+(v)$.
If both $d^+(u)$ and $d^+(v)$ are small, it may be that $d^+(u) + 1 > (1 + a) \cdot d^+(v)$ and $d^+(v) + 1 > (1 + a) \cdot d^+(u)$. Thus, when adding an edge $(u, v)$ we cannot orient the edge without violating our local condition. This significantly complicates the analysis. 

Indeed, this complication means that we can only guarantee the multiplicative condition holds \emph{between} updates to $G$, and is not maintained as an invariant as we perform updates to $G^{b}$.
Hence, to get a simple recursive algorithm to work, we have to instead work with a threshholded local condition, where we allow edges between vertices of small enough degree to get a direction in order to handle the above problem. 
We show that maintaining such a threshholded local condition is actually equivalent to maintaining the multiplicative condition between updates to $G$. 
However, working with this threshholded condition requires one to be careful. To illustrate this, we briefly sketch how the algorithms work: suppose first that every vertex has perfect information about the degrees of all other vertices. Then, it can immediately identify if incrementing/decrementing its degree causes the local condition to be violated.
If so, it can then reorient a violated edge, thus restoring its degree. 
This solves the problem for this vertex, but might move the problem to some other vertex. The key property however is that this vertex has a (significantly) smaller degree in the incremental case, or (significantly) higher degree in the decremental case. 
Hence, this cascade cannot continue many times, thus yielding a short chain. 
Every vertex, however, does not necessarily have access to the degree of all of its neighbours. 
Thus the algorithm has to supply this information somehow. 
We do so in 3 different ways: by naively informing and checking all out-neighbours of degree changes (reminiscent of the approach of Kopelowitz et al.~\cite{kopelowitz2014orienting-old}), by updating and checking estimates lazily in a round robin fashion (similar to Sawlani and Wang~\cite{sawlani2020near}) and in an amortised fashion by only checking every time a degree has changed substantially. 
The two last schemes demand that we at all times work with degree estimates that are not precise. 
Doing so is quite straightforward with a purely additive local condition, due to the simplifications mentioned earlier, but it is significantly more involved in the multiplicative case: here the conditions are very sensitive to degree changes, and so to make the analysis work, we have to be very precise about at what time a certain condition on the degree holds. Particularly so, when the degrees are small. 
This is further complicated by the fact that we now work with a threshholded local condition and thusly have to ensure that our analysis can handle all paradigms of the condition.

In Section~\ref{sec:struc}, we analyse the effect of maintaining an integral orientation that satisfies our local condition (that for all $(u, v)$,  $d(u \to v) > 0$ implies $d^+(u) \leq (1 + a) \cdot d^+(v) + c$). 
We present a general theorem showing the impact of our local condition, parametrized by $a$ and $c$.
Let $\Delta$ be the maximal out-degree in our graph. 
We immediately apply this theorem to show that for $a^{-1} \in \BigOh(\log n)$: $\Delta \in \BigOh(\rho + \log n)$ (for $c \in O(1)$) or $\Delta \in O(\rho)$ (for $c = 0$).
In Section~\ref{sec:onepluseps}, we show that choosing $c = 0$ and $a^{-1} \in O(\eps^{-2} \log n)$ allows us to maintain a factional orientation of the graph where the maximal out-degree $\Delta \leq (1 + \eps)\rho$. 
By naively rounding the fractional out-degrees, this implies that we can maintain an integral orientation where the out-degree $\Delta \leq (2 + \eps)\alpha$.

\noindent
\subsection{Parameterisation of the Algorithm}
We now introduce several components of the algorithm and analysis that can be specified to obtain various trade-offs between quality of the out-orientation and update time.
Our algorithms have the two main parameters: $\eta$ and a positive integer $b$. 
We maintain a graph $G^b$ and an orientation $\overrightarrow{G}^b$ where one of the following invariants holds: 
\begin{invariant}
\label{inv:degrees}
We maintain an orientation $\overrightarrow{G}^b$ where for every directed edge $\overrightarrow{uv}$ in $\overrightarrow{G}^b$: 
\[
d^+(u) \leq (1 + \eta \cdot b^{-1}) \cdot d^+(v). 
\]
\end{invariant}
\begin{invariant}
\label{inv:degrees_additive}
We maintain an orientation $\overrightarrow{G}^b$ where for every directed edge $\overrightarrow{uv}$ in $\overrightarrow{G}^b$: 
\[
d^+(u) \leq (1 + \eta \cdot b^{-1}) \cdot d^+(v) +2 . 
\]
\end{invariant}

Throughout the paper, we denote $\theta = 0$ if we are maintaining Invariant~\ref{inv:degrees} and $\theta = 1$ otherwise.
This way, we maintain Invariant~$\theta$ by maintaining $d^+(u) \leq (1 + \eta \cdot b^{-1}) \cdot d^+(v) + 2\theta$.
The tighter the inequalities are, the closer the maximum out-degree of the maintained out-orientation is to the maximum subgraph density.
Hence, setting $\theta = 0$ will give a better approximation than $\theta = 1$.

% For $\eta$ small enough, we note that Invariant~\ref{inv:degrees} cannot be satisfied in general (e.g., a star will require some vertex to have an outdegree of $0$), unless the number of duplicate edges $b$ is large enough. However, the number of duplicate edges required reflects poorly upon our running time. 

Note that regardless of the choice of parameters, not all graphs have an orientation that satisfies Invariant~\ref{inv:degrees}. E.g. any orientation of the graph consisting of a single edge has a directed edge $\overrightarrow{uv}$ with $1 = d^+(u) > \paren*{1+\eta\cdot b^{-1}}\cdot d^+(v) = 0$.   For convenience we will therefore need the following slightly relaxed invariant, which we show in Section~\ref{sec:basic} \emph{is} satisfiable for $\theta=0$  (and therefore for all $\theta\geq0$) as long as $0<\frac{b}{\eta}\leq\floor{\frac{b}{2}}$, i.e. when $b\geq 2$ is even and $\eta\geq 2$ or when $b\geq 3$ and $\eta \geq 3$ (or more precisely $\eta\geq\frac{2b}{b-1}$).
{
% temporarily override invariant numbering so we can do the modified invariant right
\renewcommand{\theinvariant}{$\theta^\prime$}
\begin{invariant}\label{inv:degrees_modified}
We maintain an orientation $\overrightarrow{G}^b$ where for every directed edge $\overrightarrow{uv}$ in $\overrightarrow{G}^b$: 
\[
d^+(u) \leq \max\set*{
\paren*{1 + \eta \cdot b^{-1}} \cdot d^+(v)+2\theta
,
\floor*{\frac{b}{2}}}. 
\]
    
\end{invariant}
}
The point is that for each update to $G$ this lets us do updates to $\overrightarrow{G}^b$ one edge at a time, all the while satisfying Invariant~\ref{inv:degrees_modified}, and when we are done the resulting graph satisfies Invariant~$\theta$ because of the following Lemma.

\begin{lemma}\label{lemma:inv_modified_implies_normal}
    Let $G^b$ be a graph that can be obtained from a graph $G$ by replacing each edge with $b$ copies. For all $\theta\geq 0$, any orientation $\overrightarrow{G}^b$ of $G^b$ that satisfies Invariant~\ref{inv:degrees_modified}
    also satisfies Invariant~$\theta$.
\end{lemma}
\begin{proof}
    Suppose $\overrightarrow{G}^b$ satisfies Invariant~\ref{inv:degrees_modified} and 
    let $\overrightarrow{uv}$ be any edge in $\overrightarrow{G}^b$.
    If $d^+(v)<\ceil{\frac{b}{2}}$ then (because each edge in $G^b$ is duplicated $b$ times) $d^+(u)\geq b-d^+(v)>\floor{\frac{b}{2}}$. So by Invariant~\ref{inv:degrees_modified} we have $d^+(u)\leq\paren*{1+\eta \cdot b^{-1}}\cdot d^+(v)+2\theta$ and Invariant~$\theta$ is satisfied for $\overrightarrow{uv}$.
    Otherwise $d^+(v)\geq\ceil{\frac{b}{2}}$ and by Invariant~\ref{inv:degrees_modified}, $d^+(u)\leq \max\set*{\paren*{1+\eta \cdot b^{-1}}\cdot d^+(v)+2\theta, \floor*{\frac{b}{2}}} = \paren*{1+\eta \cdot b^{-1}}\cdot d^+(v)+2\theta$ thus Invariant~$\theta$ is satisfied for $\overrightarrow{uv}$.
\end{proof}
% \paren*{1 + \eta \cdot b^{-1}} \cdot d^+(v)+2\theta
% ,
% \floor*{\frac{b}{\eta}}}. 
% \]
% and still get the Lemma. We don't want this for Section~\ref{sec:basic}, but may need it for Section~\ref{app:improved} to get the right time-dependency on $\eta$. Only relevant if $\eta$ is not a constant.  

% Wait, for $\eta\geq 2$ this is a strictly stronger invariant, so it trivially follows, and in fact we can just use the stronger invariant if we need to.  That means this is in a sense the strongest possible form of this lemma because it has the weakest requirement possible.

% At one point I thought we could shave a $\log\log n$ factor by such trickery, and it might still be possible. Won't be this iteration of the paper though.
% }

\section{A Structural Theorem}\label{sec:struc}

In this section, we formally establish the relationship between maintaining Invariant~$\theta$ for a graph $\overrightarrow{G}^b$, and the corresponding estimate of the density and arboricity of the graph $G$.
The following theorem is our result in its most general form: allowing for $(1 + \eps)$-approximations of $\rho$ and more. 
Skip ahead to Corollaries~\ref{cor:inv0}+\ref{cor:inv1} for a comprehensible application of the variables.

\begin{theorem}
\label{thm:structural}
Let $G$ be a graph and let $G^b$ be $G$ with each edge duplicated $b$ times. Let $\rho_b$ be the maximum subgraph density of $G^b$.
Let  $\overrightarrow{G}^b$ be any orientation of $G^b$ which has the following invariant: for some $c\geq 0$, every directed edge $\overrightarrow{uv}$ satisfies $d^+(u) \leq (1 + \eta \cdot b^{-1}) \cdot d^+(v)+c$.\\
Then for any $\gamma > 0$ 
there exists a value $k_{\max} \leq \log_{1 + \gamma} n$ for which:
\[
(1+\eta\cdot b^{-1})^{-k_{\max}}\Delta(\overrightarrow{G}^b) \leq (1+\gamma)\rho_b +c(\eta^{-1}\cdot b+1).
\]
\end{theorem}
\begin{proof}
Let $G^b=(V, E^b)$. We define for non-negative integers $i$ the sets: 
\[
T_i := \set*{ v\in V \suchthat d^+(v) \geq \Delta(\overrightarrow{G}^b) \cdot \paren*{1 + \eta\cdot b^{-1}}^{-i} - c\sum_{j = 1}^i \paren*{1 + \eta\cdot b^{-1}}^{-j} }
\]

\noindent
Observe that for all non-negative integers $0\leq i < j$, $T_i \subseteq T_j$.
Moreover, observe that $T_0$ contains at least one element (the element of $\overrightarrow{G}^b$ with maximum out-degree), and each $T_i$ at most $n$ elements (since they can contain at most all vertices of $G$). 
Let $k$ be the smallest integer such that $|T_{k+1}| < (1 + \gamma) |T_k|$.
It follows that $k$ is upper bounded by the value $k_{\max} = \log_{(1 + \gamma)} n$.

In order to bound the maximum out-degree of $\overrightarrow{G}^b$, we want to show that no edges can be oriented from $T_{k}$ to a vertex not in $T_{k+1}$. To do so, we assume two such candidates $u \in T_k$ and $v \not \in T_{k+1}$, and show that $\overrightarrow{uv}$ violates: $d^+(u) \leq (1 + \eta \cdot b^{-1}) \cdot d^+(v)+c$. 
Per assumption we have 
\begin{align*}
    d^+(u) &\geq (1 + \eta\cdot b^{-1})^{-k} \Delta(\overrightarrow{G}^b)  - c\sum_{j = 1}^{k} (1 + \eta\cdot b^{-1})^{-j}
    &\text{(Since $u\in T_k$)}
    \intertext{and}
    d^+(v) &< (1 + \eta\cdot b^{-1})^{-1} (1 + \eta\cdot b^{-1})^{-k} \Delta(\overrightarrow{G}^b) -  c\sum_{j = 1}^{k+1} (1 + \eta\cdot b^{-1})^{-j}.
    &\text{(Since $v\not\in T_{k+1}$)}
\end{align*}
It follows that
\begin{align*}
    (1 + \eta\cdot b^{-1}) d^+(v) + c &< (1 + \eta\cdot b^{-1})^{-k} \Delta(\overrightarrow{G}^b) - c\sum_{j = 0}^{k} (1 + \eta\cdot b^{-1})^{-j} + c \\
    &= (1 + \eta\cdot b^{-1})^{-k} \Delta(\overrightarrow{G}^b)  - c\sum_{j = 1}^{k} (1 + \eta\cdot b^{-1})^{-j} \leq d^+(u).
\end{align*}

This would violate the assumed invariant of $\overrightarrow{G}^b$.
Hence for any $u \in T_k$ and any edge $\overrightarrow{uv}$, we have $v \in T_{k+1}$ and thus: 
$
\sum_{u \in T_k} d^+(u) \leq |E^b[T_{k+1}]|.
$
Finally, we can bound the density $\rho_b$ as:
\begin{align*}
    \rho_b =
    \max_{\emptyset\subset S\subseteq V}\frac{|E^b[S]|}{|S|}
    %\rho(T_{k+1}) 
    \geq \frac{|E^b[T_{k+1}]|}{|T_{k+1}|} &\geq \frac{\sum_{u \in T_k} d^+(u) }{(1+\gamma) |T_k| } 
    \\&\geq \frac{|T_k| \cdot \paren*{ \paren*{1 + \eta\cdot b^{-1}}^{-k} \Delta\paren[\Big]{\overrightarrow{G}^b} - c\sum_{j = 1}^k \paren*{1 + \eta\cdot b^{-1}}^{-j} } }{(1+\gamma)|T_k|} .
\end{align*}
We find $\paren*{1+\gamma}\rho_b +\frac{c}{1 - \frac{1}{1 + \eta\cdot b^{-1}}}  \geq  \paren*{1 + \eta\cdot{} b^{-1}}^{-k_{\max}} \Delta\paren*{\overrightarrow{G}^b}$,
which concludes the proof.\end{proof}

The parameter $\gamma$ is needed to get a $(1+\eps)$-approximation later on, where we will require that $\gamma = \Theta(\eps)$. 
For now, one can just think of $\gamma$ as being a constant. In fact in the following corollaries, we will choose $\gamma$ so $1+\gamma = e$.

\begin{corollary}
    \label{cor:inv0}
       Denote by $\rho$ the density of $G$.
    For any $\eta$ and $b$ such that 
    $\eta b^{-1} \in \BigOh(\frac{1}{\log n})$,
    we have that Invariant~\ref{inv:degrees} for the graph $G^b$ implies:
    $\Delta(\overrightarrow{G}^b)\in \BigOh(b \rho)$ and $\Delta(\overrightarrow{G})\in \BigOh(\rho)$. 
\end{corollary}

\begin{proof} Set $\gamma=e-1$, let $k_{\max}\leq\log_{(1+\gamma)}n=\log_e n$ be as in Theorem~\ref{thm:structural}.
By our choice of $\eta$ and $b$  there exists a constant $s>0$ such that $\eta b^{-1}\leq \frac{s}{\log_e n}$ for all $n\geq 1$, thus by Theorem~\ref{thm:structural} (with $c=0$) we now have
\begin{align*}
    \Delta(\overrightarrow{G}^b) 
    &\leq
    (1+\eta\cdot b^{-1})^{k_{\max}}(1+\gamma)\rho_b
    \leq e^{\eta\cdot b^{-1}\cdot k_{\max}}(1+\gamma)\rho_b
    \leq
    e^{s+1}\rho_b
    \in\BigOh(\rho_b).
\end{align*}
Finally, we note that per definition of subgraph density, $\rho_b = b \cdot \rho$.
For all $\overrightarrow{uv}$ in $\overrightarrow{G}$, there must be at least $\ceil{\frac{b}{2}}$ edges in $\overrightarrow{G}^b$ from $u$ to $v$ (else, $\overrightarrow{G}$ would include the edge $\overrightarrow{vu}$ instead). It immediately follows that the out-degree of $u$ in $\overrightarrow{G}$ is at most $d^+(u) \cdot \ceil{\frac{b}{2}}^{-1} \in \BigOh(b^{-1} \cdot \rho_b) = \BigOh( \rho)$.
\end{proof}

% \begin{proof}
% Set $\gamma = \frac{1}{2}$. Then Theorem~\ref{thm:structural} yields $k_{\max} \leq \log_e n / \log_e (1+\gamma).$ Thus using $\log_e(1+x)\geq x/2$ whenever $x\leq 1$:
% %
% $
%     (1+\eta\cdot b^{-1} )^{-k_{\max}} = \exp(-\log_e(1+\eta b^{-1} )\cdot k_{\max}) \geq  \exp( - ( \eta b^{-1} ) \cdot k_{\max}) = \Omega(1)$.
% We now apply Theorem~\ref{thm:structural} to conclude that $\Omega(1) \Delta( \overrightarrow{G}^b) \leq  (1 + \gamma) \rho_b \leq \BigOh( \Delta( \overrightarrow{G}^b ) )$. 
% Finally,
% For all $\overrightarrow{uv}$ in $\overrightarrow{G}$, there must be at least $\frac{1}{2} b \geq \paren*{ \eta b^{-1} }^{-1}$ edges in $\overrightarrow{G}^b$ from $u$ to $v$ (else, $\overrightarrow{G}$ would include the edge $\overrightarrow{vu}$ instead). It immediately follows that the out-degree of $u$ is at most $    d^+(u) \cdot \paren*{ \eta b^{-1} } \in \BigOh(b^{-1} \cdot \rho_b) = \BigOh( \rho)$.
% \end{proof}

\begin{corollary}
    \label{cor:inv1}
           Denote by $\rho$ the density of $G$.
    Let $b = 1$ and $\eta = \frac{1}{\log_e(n)}$.
  
    Whenever Invariant~\ref{inv:degrees_additive} holds for the graph $\overrightarrow{G}=\overrightarrow{G}^b$, it must be that:  $\Delta(\overrightarrow{G})\in \BigOh(\rho+\log n)$. 
\end{corollary}

\begin{proof}
    Set $\gamma=e-1$, let $k_{\max}$ be as in Theorem~\ref{thm:structural}.
    By our choice of $\eta$ and $b$, $\eta \cdot b^{-1}=\frac{1}{\log_e n}=\frac{1}{\log_{(1+\gamma)}n}$. Thus by Theorem~\ref{thm:structural} (with $c=2$) we now have
    \begin{align*}
        \Delta(\overrightarrow{G})
        = 
        \Delta(\overrightarrow{G}^b)
        &\leq
        (1+\eta\cdot b^{-1})^{k_{\max}}
        \paren*{
        (1+\gamma)\rho_b
        +c(\eta^{-1}\cdot b+1)        
        }
        \\
        &\leq
        e^{\eta\cdot b^{-1}\cdot k_{\max}}
        \paren*{
        (1+\gamma)\rho_b
        +c(\eta^{-1}\cdot b+1)        
        }
        \\
        &\leq
        e
        \paren*{
        e\cdot\rho_b
        +c(\eta^{-1}\cdot b+1)        
        }
        \\
        &\in\BigOh(\rho+\log n)
\end{align*}
% And thus $\Delta(\overrightarrow{G})\leq \ceil{\frac{b}{2}}^{-1}\Delta(\overrightarrow{G}^b)\in \BigOh(\
\end{proof}

\section{A Simple Algorithm for Maintaining the Invariants}
\label{sec:basic}
We first provide a simple worst-case $\BigOh(\rho \log \rho \cdot \polylog(n))$  algorithm (where $\rho$ is the maximum subgraph density) to maintain Invariant~$\theta$ in $\overrightarrow{G}^b$ (i.e. we maintain one chosen invariant).  
Our data structure is purposefully more complicated than necessary here, to illustrate its use in future sections.  
Subsequent sections slightly adjust the algorithms.
Crucially, the bound on the recursive depth of our functions applies throughout the paper. 
Recall that $G^b$ is the graph $G$ with edges duplicated $b$ times. For convenience, we set $\lambda = \eta b^{-1}/64$ in the rest of the paper, and note that $(1+\lambda)^5\leq 1+\eta b^{-1}\leq 2$.
We maintain Invariant~$\theta$ using a data structure storing for all vertices $u$: 
\begin{enumerate}[(a), noitemsep]
    \item The value $d^+(u)$ of the current orientation $\overrightarrow{G}^b$,
    \item The set $N^+(u)$ in arbitrary order, and
    \item The set $N^-(u)$ in a sorted doubly linked list of buckets $B_j(u)$.
    Each bucket $B_j(u)$ contains, as a doubly linked list in arbitrary order, all $w \in N^-(u)$ where
    $ 
    j = \floor{ \log_{(1 + \lambda)} d^+(w)}$.
    The vertex $u$ has a pointer to the bucket $B_i(u)$ with $i = \floor*{ \log_{(1 + \lambda)} \max\set*{(1+\lambda)d^+(u),\floor*{\frac{b}{4}}}}$.
\end{enumerate}

\noindent
We run Algorithms~\ref{alg:insertion_in_G}+\ref{alg:deletion_in_G} on the graph $G$.
These invoke Algorithms~\ref{alg:insertion}+\ref{alg:deletion}, which in turn add  directed edges to and remove them from to $G^b$ (Algorithms~\ref{alg:add}+\ref{alg:remove}).
In our recursive algorithm calls, we may assume that for any edge insertion $(u, v)$ in $G^b$, we call $\textnormal{Insert}( \overrightarrow{uv})$ whenever $d^+(u) \leq d^+(v)$. 
Recall that $\theta$, $\eta$ and $b$ are parameters that are set beforehand:

\renewcommand\algorithmicthen{}

\noindent\begin{minipage}[t]{.5 \textwidth}
\null 
 \begin{algorithm2e}[H]
    \caption{Insert(edge $(u, v)$ in $G$)}
    \label{alg:insertion_in_G}
    \begin{algorithmic}
     %  \IF{ 
     %  $u$ and $v$ are isolated in $G$
     % }
     %  \FOR{$i \in [\frac{1}{2}b]$}
     %  \STATE Add($\overrightarrow{uv}$)
     %  \STATE Add$(\overrightarrow{vu}$)
     %  \ENDFOR
     %  \ELSE
      \FOR{$i \in [b]$ }
      \IF{ $d^+(u) \leq d^+(v)$}
      \STATE Insert($\overrightarrow{uv}$)
      \ELSE
      \STATE Insert($\overrightarrow{vu}$)
      \ENDIF
      \ENDFOR
      % \ENDIF
    \end{algorithmic}
  \end{algorithm2e}
\end{minipage}~%
\begin{minipage}[t]{.5\textwidth}
\null
 \begin{algorithm2e}[H]
    \caption{Delete(edge $(u, v)$ in $G$)}
    \label{alg:deletion_in_G}
    \begin{algorithmic}
     %    \IF{ 
     %   no other edges connect to $u$ and $v$
     % }
     %  \FOR{$i \in [\frac{1}{2}b]$}
     %  \STATE Remove($\overrightarrow{uv}$)
     %  \STATE Remove$(\overrightarrow{vu}$)
     %  \ENDFOR
     %  \ELSE
      \FOR{$i \in [b]$ }
      \IF{ $u \in N^-(v)$}
      \STATE Delete($\overrightarrow{vu}$)
      \ELSE
      \STATE Delete($\overrightarrow{uv}$)
      \ENDIF
      \ENDFOR
      % \ENDIF
    \end{algorithmic}
  \end{algorithm2e}
\end{minipage}

\noindent\begin{minipage}[t]{.5 \textwidth}
\null 
 \begin{algorithm2e}[H]
    \caption{\mbox{Insert$(\overrightarrow{uv})$, {\small where $d^+(u)\leq d^+(v)$}}}
    \label{alg:insertion}
    \begin{algorithmic}
      \STATE Add($\overrightarrow{uv}$)
      \STATE $x \gets \argmin \{ d^+(w) \mid w \in N^+(u) \}$ 
      \IF{
      % $\overrightarrow{ux}$ violates the invariant
      $
        d^+(u) > \mathrlap{\max\set*{(1 + \lambda) \cdot d^+(x) + \theta,\floor*{\frac{b}{4}}}}
      $
     }
      \STATE Remove$(\overrightarrow{ux})$ \COMMENT{restores $d^+(u)$}
      \STATE \emph{Insert($\overrightarrow{xu}$)}
      \ELSE
      \FORALL{$w \in N^+(u)$}
      \STATE Update $d^+(u)$ in Buckets($N^-(w)$) 
      \ENDFOR
      \ENDIF
    \end{algorithmic}
  \end{algorithm2e}
\end{minipage}~%
\begin{minipage}[t]{.5\textwidth}
\null
 \begin{algorithm2e}[H]
    \caption{Delete$(\overrightarrow{uv})$}
    \label{alg:deletion}
    \begin{algorithmic}
    \STATE  Remove($\overrightarrow{uv}$)
      \STATE $x \gets \textnormal{First( Max( Buckets($ N^-(u)$)))} $
      \IF{
      % $\overrightarrow{xu}$ violates the invariant
      $d^+(x) > \mathrlap{\max\set*{(1 + \lambda ) \cdot d^+(u)   + \theta,\floor*{\frac{b}{4}}}}
      $
}
            \STATE Add$(\overrightarrow{ux})$ \COMMENT{restores $d^+(u)$}
            \STATE \emph{Delete($\overrightarrow{xu}$)}
      \ELSE
      \FORALL{$w \in N^+(u)$}
      \STATE Update $d^+(u)$ in Buckets($N^-(w)$) 
      \ENDFOR
      \ENDIF
    \end{algorithmic}
  \end{algorithm2e}
\end{minipage}

\noindent\begin{minipage}[t]{.5 \textwidth}
\null 
 \begin{algorithm2e}[H]
    \caption{\mbox{Add$(\overrightarrow{uv})$}}
    \label{alg:add}
    \begin{algorithmic}
      \STATE $d^+(u) = d^+(u) + 1$
      \IF{ 
      $
        \nexists \overrightarrow{uv} \in \overrightarrow{G}^b
      $
     }
      \STATE Add $u$ to $N^-(v)$ and $v$ to $N^+(u)$
      \ENDIF
      \STATE Add one edge $\overrightarrow{uv}$ to $\overrightarrow{G}^b$.
    \end{algorithmic}
  \end{algorithm2e}
\end{minipage}~%
\begin{minipage}[t]{.5\textwidth}
\null
 \begin{algorithm2e}[H]
    \caption{Remove$(\overrightarrow{uv})$}
    \label{alg:remove}
    \begin{algorithmic}
    \STATE  $d^+(u) = d^+(u) - 1$
             \STATE Remove one edge $\overrightarrow{uv}$ from $\overrightarrow{G}^b$.
      \IF{ 
      $
        \nexists \overrightarrow{uv} \in \overrightarrow{G}^b
      $
     }
      \STATE Remove $u$ from $N^-(v)$, $v$ from $N^+(u)$
      \ENDIF
    \end{algorithmic}
  \end{algorithm2e}
\end{minipage}

\begin{definition}
\label{def:time}
We count time in discrete steps.
%starting from 0
 A new time step starts just before Algorithm~\ref{alg:insertion_in_G} calls Insert or Algorithm~\ref{alg:deletion_in_G} calls Delete.  
For a time $t$, we denote for any variable $\phi$ in our code by $\phi_t$ its value before the invoking insertions (or deletions) at time $t$.
E.g., for a vertex $w$, $d^+(w)_t$ is the out-degree before invoking insertions at time $t$, and $d^+(w)_{t+1}$ is the out-degree just after.
\end{definition}

\subsection{Maintaining Invariant~\ref{inv:degrees_additive}.}
We show that by setting $\theta = 1$ (and choosing $\eta$ and $b$ carefully) we maintain  Invariant~\ref{inv:degrees_additive}:
    
\begin{theorem}
\label{thm:invariant1}
Let $G$ be a dynamic graph and $\rho$ be the density of $G$ at time $t$. 
We can choose our variables $\theta = 1$, $b = 1$ and $\eta \in \Theta(\log n)$ to maintain an out-orientation $\overrightarrow{G}^b = \overrightarrow{G}$ in $O\paren*{ (\rho + \log n) \cdot   \log n \cdot \log \rho }$ time per update in $G$ such that  Invariant~\ref{inv:degrees_additive} holds for $\overrightarrow{G}$. Moreover:

\begin{itemize}[noitemsep, nolistsep]
 \item $\forall u$, the out-degree $d^+(u)_{t+1}$ in $\overrightarrow{G}$ is at most $\BigOh( \rho + \log n)$, \quad \quad \small{(i.e. $\Delta( \overrightarrow{G}) \in \BigOh( \rho + \log n)$)}
\end{itemize}
\end{theorem}

\begin{proof}
Invariant~\ref{inv:degrees_additive} demands that $\forall \overrightarrow{uv}$ we maintain $d^+(u) < (1 + \eta b^{-1}) d^+(v) + 2$.
Corollary~\ref{cor:inv1} implies that if after time $t$ we maintain Invariant~\ref{inv:degrees_additive}, then we obtain the desired upper bound on all $d^+(u)_{t+1}$.  
What remains is to show that our algorithms maintain Invariant~\ref{inv:degrees_additive} in the desired runtime. We do this in three steps as we show: 

\begin{description}[noitemsep]
    \item[Correctness:]  Our algorithms maintain Invariant~\ref{inv:degrees_additive} in $G^b$,
        \item[Recursive depth:] Algorithms~\ref{alg:insertion}+\ref{alg:deletion} have a recursive depth of $\BigOh\paren*{\lambda^{-1} \cdot \log \rho}$, and 
    \item[Time:]  Algorithms~\ref{alg:insertion}+\ref{alg:deletion} spend $\BigOh(\rho + \log n)$ time before entering the recursion.
\end{description}

We prove these three properties for deletions only. 
Invoking Delete$( \overrightarrow{x_0 v})$ may cause us to recursively invoke Delete$(\overrightarrow{x_{i +1 } x_{i} })$: flipping a backward chain in $G^b$ from $x_0$.  
Only the final vertex $x_f$ in this chain decreases its out-degree once we terminate. 
For insertions we flip a forward chain $\overrightarrow{x_i x_{i+1}}$, which is handled symmetrically.

\subparagraph{Correctness.}
We show that we maintain Invariant~\ref{inv:degrees_additive}.
Suppose that we terminate at a vertex $x_f$. 
Then after our sequence of flips, the vertex $x_f$ is the only vertex that changed its out-degree (i.e. only for $x_f$: $d^+(x_f)_{t + 1} = d^+(x_f)_{t} - 1$). 
Because our algorithm terminated and $b = 1$, for $x \gets $ First( Max( Buckets($N^-(x_f)$))), $d^+(x)_t \leq  \max \{ (1 + \lambda) (d^+(x_f)_t - 1) + \theta, \floor{\frac{b}{4}} \}$. 
For all $w \in N^-(x_f)$: $d^+(w)_{t+1} \leq (1 + \lambda) d^+(x)_{t+1}$.
It follows $d^+(w)_{t+1} \leq \max\{  (1 + \lambda)^2 d^+(x_f)_{t+1} + 2, \floor{\frac{b}{2}} \}$.
We may apply Lemma~\ref{lemma:inv_modified_implies_normal} to conclude that,  once terminated, we satisfy Invariant~\ref{inv:degrees_additive}.

\subparagraph{Recursive depth. }
What remains is to upper bound the recursive depth of our algorithm, proving termination.
Our code implies that for all $i$: $d^+(x_{i+1})_t > (1 + \lambda) ( d^+(x_{i})_t - 1) + \theta$. Thus $d^+(x_{i+1})_t \geq  d^+(x_{i})_t + 1$. 
Let $x_s$ be the last vertex in the chain where $d^+(s)_t \in \BigOh( \log n)$. 
 The fact that out-degrees are integer and strictly increasing along the backward chain, implies that there are $ \BigOh(\log n)$ vertices preceding $s$. 
If $f = s$, the recursive depth is $ \BigOh(\log n) = \BigOh(\lambda^{-1})$ per definition.

Otherwise, we note that before this sequence of updates, we satisfied Invariant~\ref{inv:degrees_additive} and thus (by Corollary~\ref{cor:inv1}) know that for all $i$: $d^+(x_i)_t \in  \BigOh(\rho + \log n)$. 
If there exist vertices $x_i$ with $i > s$, then $\rho \in \Omega(\log n)$ and thus  $ \BigOh(\rho + \log n) =  \BigOh(\rho)$. 
Now we consider all $i > s$. 
We know that $d^+(x_{i+1})_t > (1 + \lambda) ( d^+(x_{i})_t - 1)$.
Thus, (using $d^+(x_{i})_t \geq  d^+(x_{i-1})_t + 1$) we get that: $d^+(x_{i+1})_t >  (1 + \lambda) d^+(x_{i-1})_t$. 
It follows that there are at most $  \log_{ (1 + \lambda)} \BigOh(\rho) =  \BigOh( \lambda^{-1} \log \rho)$ vertices in the chain of flipped edges: which upper bounds our recursive depth.

\subparagraph{Time spent.}
Whenever we insert a vertex $v \in N^{-}(u)$, it is either because we added the edge $(u, v)$ to $G$ (occurring once) or, because we flipped an edge $\overrightarrow{uv}$.
In the first case, we may afford spending
$\BigOh(\log_{(1+\lambda)}d^+(v)) = \BigOh(\lambda^{-1}\log(b\rho + \log n)) = \BigOh(\lambda^{-1}\log n)$ time searching through all buckets for the bucket containing $v$. 
In the latter case, for Insert we have $\max\set*{(1+\lambda)d^+(u)_t,\floor{\frac{b}{4}}}<d^+(v)_t+1\leq\max\set*{(1+\lambda)d^+(u)_t,\floor{\frac{b}{4}}}+1$, and for Delete we have $\max\set*{(1+\lambda)(d^+(u)_t-1),\floor{\frac{b}{4}}}<d^+(v)_t<(1+\lambda)\max\set*{(1+\lambda)(d^+(u)_t-1),\floor{\frac{b}{4}}}$.
Using the pointer from $u$ to the bucket $B_i(u)$ where $i=\floor*{\log_{(1+\lambda)}\max\set*{(1+\lambda)d^+(u),\floor*{\frac{b}{4}}}}$, we may insert $v$ into the correct bucket in $\BigOh(1)$ time. 
%In the latter case, $d^+(v)_t \in \BigOh((1 + \lambda)d
For each call of Delete$(\overrightarrow{x_{i+1} x_i})$, we spend $ \BigOh(1)$ time retrieving the vertex $x$ before we recurse. 
For the vertex $x_f$ at the end of the recursion, we consider all $ \BigOh(\rho + \log n)$ vertices $w \in N^+(x_f)$. We update the bucket that $x_f$ is in. 
Denote by $r(x_f)_{t+1} = \floor{ \log_{(1 + \lambda)} d^+(x_f)_t}$ the \emph{rank} of $f$ (i.e., the index of each bucket contained $x_f$ at time $t$). 
The rank of $x_f$ changes by at most $1$, hence we may update our data structure in $ \BigOh(d^+(f)_{t+1})$ time. 

For each call of Insert$(\overrightarrow{x_{i+1} x_i})$, we spend $ \BigOh(d^+(x_{i+1})_t) =  \BigOh(\rho + \log n)$ time retrieving the vertex $x_{i+2}$ before we recurse. Updating the data structure again takes $ \BigOh(1)$ time per updated element. 
It follows that the total time spent adding or removing an arc in $G^b$ is $ \BigOh( (\rho + \log n) \cdot \lambda^{-1} \log \rho) =  \BigOh((\rho + \log n) \log n \log \rho)$. 
Since $b = 1$, the theorem follows. 
\end{proof}

\subsection{Maintaining Invariant~\ref{inv:degrees}.}
We show that by setting $\theta = 0$ (and choosing $\eta$ and $b$ carefully) we maintain Invariant~\ref{inv:degrees}:

\begin{theorem}
\label{thm:invariant0}
Let $G$ be a dynamic graph and $\rho$ be the density of $G$ at time $t$.
We can choose our variables $\theta = 0$,  $\eta=3$, and $b \in \Theta(\log n), b\geq 2$ to maintain an out-orientation $\overrightarrow{G}^b$ in $\BigOh\paren*{ b \cdot \rho \cdot  \lambda^{-1}  \cdot \log \rho } = \BigOh( \rho \cdot \log^2 n \log \rho)$ time per update in $G$, maintaining Invariant~\ref{inv:degrees} for $\overrightarrow{G}^b$ with:
\begin{itemize}[noitemsep, nolistsep]
    \item $\forall v$, the out-degree $d^+(v)_{t+1}$ in $\overrightarrow{G}^b$ is at most $\BigOh( b  \cdot \rho)$, and
    \item $\forall u$, the out-degree of $u$ in $\overrightarrow{G}$ is at most $\BigOh(\rho)$.
\end{itemize}
\end{theorem}

\begin{proof}
We show that at all times we maintain Invariant~\ref{inv:degrees_modified} for $\theta=0$.
Corollary~\ref{cor:inv0}, and our choice of variables, implies the desired upper bound on the out-degree of each vertex. 
We again consider:

\begin{description}[noitemsep]
    \item[Correctness:]  Our algorithms maintain Invariant~\ref{inv:degrees_modified} in $G^b$,
        \item[Recursive depth:] Algorithms~\ref{alg:insertion}+\ref{alg:deletion} have a recursive depth of $\BigOh\paren*{\lambda^{-1} \cdot \log \rho}$, and 
    \item[Time:]  Algorithms~\ref{alg:insertion}+\ref{alg:deletion} spend $\BigOh(\rho)$ time before entering the recursion.
\end{description}

We show the proof for deletions. 
Again, the proof for insertions is symmetrical (flipping a forward chain). 
Invoking Delete$( \overrightarrow{x_0 v})$ may cause us to recursively invoke Delete$(\overrightarrow{x_{i +1 } x_{i} })$: flipping a backward chain in $G^b$ from $x_0$.  
Only the last vertex $x_f$ in this chain decreases its out-degree once we terminate.

\subparagraph{Correctness.}
Suppose that we terminate at a vertex $x_f$. 
Then after our sequence of flips, only the vertex $x_f$ changed its out-degree (i.e. only for $x_f$: $d^+(x_f)_{t+1} = d^+(x_f)_{t} - 1$). 
Because our algorithm terminated, for $x \gets $ First( Max( Buckets($N^-(x_f)$))) it must be that:  
$d^+(x)_t \leq  
\max \{ (1 + \lambda) (d^+(x_f)_{t} - 1 ) + \theta, \floor{ \frac{b}{4} }  \}.$ 
It follows that for all vertices $w \in N^-(x_f):$ $d^+(w)_{t} \leq (1 + \lambda) \max \{ (1 + \lambda) (d^+(x_f)_{t} - 1 ) + \theta, \floor{ \frac{b}{4} }  \}$. 
Substituting $d^+(x_f)_t$ for $d^+(x_f)_{t+1}$ and using that by our choice of parameters, $(1+\lambda)^2\leq 1+\eta b^{-1}\leq 2$ now gives that for all $w \in N^-(x_f):$
$d^+(w)_{t} \leq  \max \{ (1 + \eta b^{-1}) d^+(x_f)_{t+1}  + 2\theta, \floor{ \frac{b}{2} }  \}$. 
%$d^+(w)_{t} \leq  \max \{ (1 + \eta b^{-1}) d^+(x_f)_{t+1}  + 2\theta, \floor{ \frac{b}{2} }  \}$. 
By Lemma~\ref{lemma:inv_modified_implies_normal}, this implies that we maintain Invariant~\ref{inv:degrees}.

\subparagraph{Recursive depth.}
What remains is to upper bound the recursive depth of our algorithm.
Let $x_s$ be the first vertex in the chain where $d^+(x_s)_t \geq\floor{\frac{b}{4}}+2$. 
Note that per definition of our algorithm, $d^+(x_1)_t\geq\floor{\frac{b}{4}}+1$.
Thus, for all $i \geq 1$:  $ d^+(x_{i+1})_{t} \geq d^+(x_i)_t + 1$ and $s \leq 2$.  
We now make a case distinction. 
If $f \in \BigOh(\lambda^{-1})$ then per definition, the recursive depth is $\BigOh(\lambda^{-1})$. 

Otherwise, for all $i > 2$ it must be that $d^+(x_{i+1})_t > (1 + \lambda) ( d^+(x_{i}) - 1) \geq  (1 + \lambda) d^+(x_{i-1})$. 
Before this sequence of updates, we satisfied Invariant~\ref{inv:degrees}. Thus, by Corollary~\ref{cor:inv0}, for all $i$, $d^+(x_i)_t \in \BigOh(\rho)$.

\subparagraph{Time spent.}
The proof upper bounding the time spent is identical to that of Theorem~\ref{thm:invariant0}.
The one exception being, that the out-degree $d^+(u)$ for all vertices $u$ is at most $\rho$. 
Thus, the running time per update in $G^b$ is $\BigOh(\rho  \cdot \lambda^{-1} \log \rho) =  \BigOh(\rho \log n \log \rho)$. 
Each update in $G$ triggers $\Theta(\log n)$ updates in $G^b$ and so the runtime follows. 
%For each call of Delete$(\overrightarrow{x_{i+1} x_i})$, we spend $\BigOh(1)$ time retrieving the vertex $x$ before we recurse. 
%For the vertex $x_f$ at the end of the recursion, we consider all $\BigOh(\rho)$ vertices $w \in N^+(x_f)$. We update the bucket that $x_f$ is in. 
%Denote by $r(x_f)_{t+1} = \floor{ \log_{(1 + \lambda)} d^+(x_f)_t}$ the \emph{rank} of $f$ (i.e., the index of each bucket contained $x_f$ at time $t$). 
%The rank of $x_f$ changes by at most $1$, hence we may update our data structure in $\BigOh(1)$ time. 

%For each call of Insert$(\overrightarrow{x_{i+1} x_i})$, we spend $\BigOh(d^+(x_{i+1})_t) = \BigOh(\rho)$ time retrieving the vertex $x_{i+2}$ before we recurse. Updating the data structure again takes $\BigOh(1)$ time per updated element. 
%It follows that the total time spent adding or removing an arc in $G^b$ is $\BigOh( \rho \cdot \lambda^{-1} \log \rho) = \BigOh(\rho \log n \log \rho)$. 
%For each operation in $G$, we perform $b$ non-recursive calls to Insert or Delete in $G^b$. The fact that $b \in \Theta(\log n)$ implies the theorem. 
\end{proof}

\section{Improved worst case algorithms}
\label{app:improved-rank}
We adapt the algorithm of Section~\ref{sec:basic}, replacing the algorithms for inserting and deleting directed edges in $G^b$ to update our running time. 
We store  $u \in N^-(v)$ in buckets determined not by the actual out-degrees $d^+(u)$ but rather by an approximation of what we call the \emph{out-rank} $r(u) = \floor*{\log_{(1+\lambda)}d^+(u)}$.

\begin{definition}
    For each vertex $v$, for all vertices $u \in N^{-}(v)$, we define the \emph{perceived} out-rank $r_v(u)$ as some integer stored in $v$ for $u \in N^-(v)$ (which we show is at most $1$ removed from $r(u)$).
\end{definition}

\noindent

In this section, 
We maintain for all $u$:
\begin{enumerate}[(a), noitemsep]
    \item The \emph{exact} value $d^+(u)$ of the current orientation $\overrightarrow{G}^b$,
    \item The set $N^+(u)$ in a linked list and a pointer some current `position' in the linked list. 
    \item The set $N^-(u)$ in a doubly linked list of buckets $B_j(u)$ sorted by $j$ from high to low.
    Each bucket $B_j(u)$ contains, as a doubly linked list in arbitrary order, all $w \in N^-(u)$ where
    $r_u(w)=j$.
    The vertex $u$ has a pointer to the bucket $B_i(u)$ with $i = \floor*{ \log_{(1 + \lambda)} \max\set*{(1+\lambda)d^+(u),\floor*{\frac{b}{4}}}}$.
\end{enumerate}

Any update in $G$ invokes Algorithms~\ref{alg:insertion_in_G}+\ref{alg:deletion_in_G}.
These algorithms now  invoke Algorithm~\ref{alg:insertion_rank} or \ref{alg:deletion_rank} (instead of~\ref{alg:insertion} or \ref{alg:deletion}).
These two in turn invoke the normal add and remove functions (Algorithm~\ref{alg:add}+\ref{alg:remove}). 
Whenever we add a vertex $w$ to a set $N^-(u)$, we set $r_u(w) = r(w)$. And when we add a vertex $v$ to a set $N^+(u)$, we do so in the position immediately \emph{before} the current position, so it becomes the last one we visit when we round-robin over $N^+(u)$.

\noindent\begin{minipage}[t]{.5 \textwidth}
\null 
 \begin{algorithm2e}[H]
    \caption{Insert($\overrightarrow{uv}$)  }
    \label{alg:insertion_rank}
    \begin{algorithmic}
    \STATE Add($\overrightarrow{uv}$)
    \FOR{$x$ in next $\ceil{\frac{2}{\lambda}}$ neighbours in $N^+(u)$} 
    %\STATE $x \gets \textnormal{First( First( LinkedList of $ N^+(u)$))} $
      \IF{ 
      $
      d^+(u) > \mathrlap{\max\set*{(1 + \lambda) \cdot d^+(x) + \theta,\floor*{\frac{b}{4}}}}
      $
     }
      \STATE Remove$(\overrightarrow{ux})$
      \STATE \emph{Insert($\overrightarrow{xu}$)}
      \STATE \textsc{Break}
      % \RETURN
      \ENDIF
      \ENDFOR
    \FORALL {
      $x$ visited in the previous loop
    }
    \STATE $r_x(u) = r(u)$
    \STATE Move $u$ to bucket $B_{r_x(u)}$ in $N^-(x)$
    \ENDFOR

    \end{algorithmic}
  \end{algorithm2e}
\end{minipage}~%
\begin{minipage}[t]{.5\textwidth}
\null
 \begin{algorithm2e}[H]
    \caption{Delete($\overrightarrow{uv}$)}
    \label{alg:deletion_rank}
    \begin{algorithmic}
        \STATE Remove($\overrightarrow{uv}$)
      \STATE $x \gets \textnormal{First( Max( Bucket($N^-(u)$)))} $
      % \STATE $i \gets \max\set*{j\suchthat B_j \in N^-(v) \text{ is nonempty }}$
      \IF{
      $
       d^+(x) > \mathrlap{\max\set*{(1 + \lambda) \cdot d^+(u)   +  \theta,\floor*{\frac{b}{4}}}}
      $
}
            \STATE Add$(\overrightarrow{ux})$
            \STATE \emph{Delete($\overrightarrow{xu}$})
      \ELSE
    \FOR {$w$ in next $\ceil{\frac{2}{\lambda}}$ neighbors in $N^+(u)$ }
    \STATE $r_w(u) = r(u)$
    \STATE Move $u$ to bucket $B_{r_w(u)}$ in $N^-(w)$
    \ENDFOR
      \ENDIF
    \end{algorithmic}
  \end{algorithm2e}
\end{minipage}

\paragraph{Overview of techniques.}
Note that after incrementing (or decrementing) $d^+(u)$, we flip an edge $\overrightarrow{ux}$ (or $\overrightarrow{xu}$) whenever the following conditions hold:
\begin{align*}
    d^+(u) &> \max\set*{(1 + \lambda)  \cdot d^+(x) + \theta, \floor*{\frac{b}{4}}} & \text{(for Insert)}
    \\
    d^+(x) &> \max\set*{(1 + \lambda) \cdot d^+(u) + \theta, \floor*{\frac{b}{4}}} & \text{(for Delete)}.    
\end{align*}
These checks are the same as in Section~\ref{sec:basic} (and Section~\ref{sec:amortized} for deletions). 
As a result, the recursive depth of our algorithm is identical to that of Section~\ref{sec:basic}. 

The big difference with Section~\ref{sec:basic}, is that during insertions we do not loop over all $x\in N^+(u)$ each time (as that would be too expensive). Instead, we do a round robin scheme where we rely on the fact that if we recently checked the condition for edge $\overrightarrow{ux}$ without flipping it, then we need to add many more outgoing edges from $u$ before it violates the actual Invariant~\ref{inv:degrees_modified}. 
By checking $\ceil{\frac{2}{\lambda}}$ edges each time in round-robin order we are guaranteed to revisit $\overrightarrow{ux}$ before that happens.  

The second difference with Section~\ref{sec:basic}, is that for each vertex $u$ we cannot store a data structure on the in-neighbors of $u$ that uses their actual out-degree. 
Instead, we bucket the vertices $x \in N^-(u)$ using their out-degree at the time of adding the arc $\overrightarrow{xu}$. 
The location of $x$ in this data structure is thereby its \emph{perceived} rank $r_u(x)$.
Whenever we insert or delete an arc in $G^b$, we get a recursive call to our insertion and deletion functions that flips a chain of edges.
Only the final vertex $x_f$ on this chain changes their actual out-degree. 
Hence, for this final vertex $x_f$, we perform round robin over the $\ceil*{\frac{2}{\lambda}}$ next $w \in N^+(x_f)$ to update the perceived rank of $x_f$ in $N^-(w)$. Again, we can not afford to update all of them.

Recall that we parametrized time according to Definition~\ref{def:time}. We show:

\begin{lemma}\label{lemma:rank_bounded_change}
    Let $r_v(u)$ get updated by an Insert or Delete at time $s$.
    Let the next update to $r_v(u)$ occur during an Insert or Delete at time $t$. Then
    $\abs{d^+(u)_t-d^+(u)_s} \leq \tfrac{\lambda}{2}d^+(u)_s$ and $\abs{r(u)_t-r(u)_s}\leq 1$.
\end{lemma}
\begin{proof}
    Only out-neighbours to $u$ that exist at time $s$ can be visited by the round-robin procedure before $r_v(u)$ is updated again.
    Since we visit $\ceil{\frac{2}{\lambda}}$ of them per Insert or Delete that changes $d^+(u)$, we can do at most $d^+(u)_s / \ceil{\frac{2}{\lambda}} \leq \frac{\lambda}{2}d^+(u)_s$ Inserts or Deletes changing $d^+(u)$ before time $t$. 
    Thus, since $0< \lambda < 1$:
    \begin{align*}
        d^+(u)_t &\geq \paren*{1-\tfrac{\lambda}{2}}d^+(u)_s > \paren*{1+\lambda}^{-1}d^+(u)_s &\implies r(u)_t \geq r(u)_s-1 
        \\
        d^+(u)_t &\leq \paren*{1+\tfrac{\lambda}{2}}d^+(u)_s < \paren*{1+\lambda}d^+(u)_s &\implies r(u)_t \leq r(u)_s+1
        &
    \end{align*}
\end{proof}

\begin{lemma}\label{lemma:perceived_rank_error}
    For all edges $\overrightarrow{uv}$ at all steps during Insert or Delete, $\abs{r_v(u)-r(u)}\leq 1$.
\end{lemma}
\begin{proof}
    Follows trivially from Lemma~\ref{lemma:rank_bounded_change} by the fact that each time it gets updated the true value has changed by at most $1$.
\end{proof}

We now apply an argument that we have applied in previous sections, introducing a bit more slack than previously:

\begin{lemma}\label{lemma:delete_rank}
   During a Delete$(\overrightarrow{uv})$ at time $t$, let $x \gets $ First( Max($N^-(u)$)) and 
    \[
    d^+(x)_t \leq \max\set*{\paren*{1+\lambda}(d^+(u)_t-1)+\theta, 
    \floor*{\frac{b}{4}}}.
    \]
    Then for all $w \in N^-(u)$ it must be that: 
     \begin{align*}
        d^+(w)_t
        &\leq
    (1+\lambda)^3\cdot\max\set*{\paren*{1+\lambda}(d^+(u)_t-1)+\theta, \floor*{\frac{b}{4}}}
        \\
        &\leq  
    \max\set*{\paren*{1+\eta b^{-1}}(d^+(u)_t-1)+2\theta, \floor*{\frac{b}{2}}}.
     \end{align*}
\end{lemma}
\begin{proof}
    The vertex $x \gets $ First( Max($N^-(u)$)) has the largest perceived rank of all vertices in $N^-(u)$.  
    Thus, the perceived rank $r_u(w)_t$ is at most $r_u(x)_t$. 
    By Lemma~\ref{lemma:perceived_rank_error}, we now get: 
    
    \[
    r(x)_t \geq r_u(x)_t-1 \geq r_u(w)_t-1 \geq r(w)_t-2 \implies d^+(x)_t\geq (1+\lambda)^{r(x)_t}\geq (1+\lambda)^{r(w)_t-2}\geq (1+\lambda)^{-3}d^+(w)_t.
\]

\noindent
It follows that $d^+(w)_t \leq (1 + \lambda)^3 \cdot \max\set*{\paren*{1+\lambda}(d^+(u)_t-1)+\theta, 
    \floor*{\frac{b}{4}}}$.
    By noting that $(1 + \lambda)^5 \leq (1 + \eta b^{-1})\leq 2$ we recover the lemma. 
\end{proof}

\begin{lemma}\label{lemma:insert_rank}
    If during an Insert at time $s$, the out-neighbour $x\in N^+(u)_s$ is verified to satisfy $d^+(u)_s+1\leq \max\set*{(1+ \lambda)d^+(x)_s + \theta, \floor{\frac{b}{4}}}$, then at any time $t$ up to and including the next time that we check the constraint we have that:
    \begin{align*}
    d^+(u)_t 
    &\leq (1+\lambda)^4\cdot\max\set*{(1+\lambda)d^+(x)_t+\theta, \floor*{\frac{b}{4}}}
    \\
    &\leq\max\set*{(1+\eta b^{-1})d^+(x)_t+2\theta, \floor*{\frac{b}{2}}}.        
    \end{align*}
\end{lemma}
\begin{proof}
    If there are no Deletes changing $d^+(x)$ between times $s$ and $t$, we have $d^+(x)_s \leq d^+(x)_t$ and
    \begin{align*}
        d^+(u)_t
        &\leq (1+\lambda)\cdot d^+(u)_s   
        &\text{(By Lemma~\ref{lemma:rank_bounded_change})}
        \\
        &\leq (1+\lambda)\cdot (d^+(u)_s + 1)
        \\
        &\leq (1+\lambda)\cdot
        \max\set*{\paren*{1+\lambda}d^+(x)_s + \theta, \floor*{\frac{b}{4}}}
        &\text{(By our assumption)}
        \\
        &\leq (1+\lambda)\cdot
        \max\set*{\paren*{1+\lambda}d^+(x)_t + \theta, \floor*{\frac{b}{4}}}
        &\text{(Since $d^+(x)_s \leq d^+(x)_t$)}
        \\
        &\leq \max\set*{(1+\eta b^{-1})d^+(x)_t+2\theta, \floor*{\frac{b}{2}}}.
    \end{align*}
    Suppose now that there was a Delete between times $s$ and $t$ that changed $d^+(x)$. Denote by $s'$ the time just after the last such delete finished. 
    It must be that $s<s'\leq t$.  
    Then $d^+(x)_{s'-1}-1=d^+(x)_{s'}$. Since after $s'$, there was no deletion decreasing $d^+(x)$ it must be that $d^+(x)_{s'}  \leq d^+(x)_t$ and by Lemma~\ref{lemma:delete_rank}, for all $w\in N^-(x)_{s'}$:
    \begin{align*}
        d^+(w)_{s'} &\leq (1+\lambda)^3\cdot
        \max\set*{\paren*{1+\lambda}d^+(x)_{s'}+\theta, \floor*{\frac{b}{4}}}
    \end{align*}
    In particular, $u\in N^-(x)_{s'}$ and
    \begin{align*}
        d^+(u)_t
        &\leq
        (1+\lambda)\cdot d^+(u)_{s'}
        &\text{(By Lemma~\ref{lemma:rank_bounded_change})}
        \\
        &\leq
        (1+\lambda)^4\cdot\max\set*{\paren*{1+\lambda}d^+(x)_{s'}+\theta, \floor*{\frac{b}{4}}}
        &\text{(By Lemma~\ref{lemma:delete_rank})}
        \\
        &\leq
        (1+\lambda)^4\cdot\max\set*{\paren*{1+\lambda}d^+(x)_{t}+\theta, \floor*{\frac{b}{4}}}
        &\text{(Since $d^+(x)_{s'} \leq d^+(x)_{t}$)}
        \\
        & \leq
        \max\set*{\paren*{1+\eta b^{-1}}d^+(x)_{t}+2\theta, \floor*{\frac{b}{2}}}
        &&
    \end{align*}
\end{proof}

\begin{lemma}[Maintaining Invariant~\ref{inv:degrees_modified}]
\label{lemma:rank_invariant_maintenance}
Whenever Algorithms~\ref{alg:insertion_rank} and~\ref{alg:deletion_rank} terminate, they maintain an orientation $\overrightarrow{G}^b$ where for each edge $\overrightarrow{uv}$ in $\overrightarrow{G}^b$, $
d^+(u) \leq \max\set*{(1 + \eta b^{-1} ) \cdot  d^+(v) + 2 \theta, \floor*{\frac{b}{2}}}$.
\end{lemma}
\begin{proof}
    By construction, when calling Insert$(\overrightarrow{uv})$ at time $t$ we always have $d^+(u)_t \leq d^+(v)_t$. As argued in Theorems~\ref{thm:invariant0}+~\ref{thm:invariant1}s, this new edge may never invalidate Invariant~\ref{inv:degrees_modified} between $u$ and $v$. 
    Now consider the chain of edges that get recursively flipped until we reach the final vertex $x_f$. 
    The vertex $x_f$ is the only vertex for which $d^+(x_f)_{t+1} = d^+(x_f)_t + 1$.
    Thus, for all other vertex pairs not including $x_f$, Invariant~\ref{inv:degrees_modified} is maintained. 
    Since the algorithm terminated at $x_f$ it must be that for all $x \in N^+(x_f)_{t}$ where the constraint was checked at time $t$: 
    $d^+(x_f)_{t+1} = d^+(x_f)_t + 1 \leq \max\set*{(1+ \lambda)d^+(x)_t + \theta, \floor{\frac{b}{4}}} = \max\set*{(1+ \lambda)d^+(x)_{t+1} + \theta, \floor{\frac{b}{4}}}$.
    By Lemma~\ref{lemma:insert_rank}, it follows that Invariant~\ref{inv:degrees_modified} is maintained between $x_f$ and all vertices in $N^+(x_f)_{t+1}$.

    The argument for Delete$(\overrightarrow{uv})$ is symmetrical, applying Lemma~\ref{lemma:delete_rank} instead.
%    Suppose now that a Delete$(\overrightarrow{uv})$ at time $t$ is causing the violation for some edge $\overrightarrow{x'u}$.  This can only happen when the Delete does not recurse, because the recursive case restores $d^+(u)$ before recursing.
  %  But if the Delete at time $t$ does not recurse and causes a violation at time $t+1$ for edge $\overrightarrow{x'u}$ then
 %   \begin{align*}
 %       d^+(x')_t = d^+(x')_{t+1} 
 %       &> \max\set*{\paren*{1+\eta b^{-1}}d^+(u)_{t+1} + 2\theta, \floor*{\frac{b}{2}}} 
 %       \\
   %     &= \max\set*{\paren*{1+\eta b^{-1}}(d^+(u)_t-1) + 2\theta, \floor*{\frac{b}{2}}} 
  %      \\
  %      &> (1+\beta)^{-1}\max\set*{\paren*{1+\eta b^{-1}}(d^+(u)_t-1) + 2\theta, \floor*%{\frac{b}{2}}} 
  %  \end{align*}
\end{proof}

\begin{lemma}
    \label{lem:timebeforerecurse}
    Algorithms~\ref{alg:insertion_rank}+\ref{alg:deletion_rank} spend $\BigOh(\lambda^{-1})$ time before recursing, except for the outermost call which spends $\BigOh(\lambda^{-1}\log n)$ time.
\end{lemma}

\begin{proof}
    Whenever we insert a vertex $v \in N^{-}(u)$, it is either because we added the edge $(u, v)$ to $G$ (occurring once) or, because we flipped an edge $\overrightarrow{uv}$.
In the first case, we may afford spending
$\BigOh(\log_{(1+\lambda)}d^+(v)) = \BigOh(\lambda^{-1}\log(b\rho + \log n)) = \BigOh(\lambda^{-1}\log n)$ time searching through all buckets for the bucket containing $v$. 
In the latter case, for Insert we have $\max\set*{(1+\lambda)d^+(u)_t,\floor{\frac{b}{4}}}<d^+(v)_t+1\leq(1+\lambda)^4\max\set*{(1+\lambda)d^+(u)_t,\floor{\frac{b}{4}}}+1$ by Lemma~\ref{lemma:insert_rank}, and similarly for Delete we have $\max\set*{(1+\lambda)(d^+(u)_t-1),\floor{\frac{b}{4}}}<d^+(v)_t\leq(1+\lambda)^3\max\set*{(1+\lambda)(d^+(u)_t-1),\floor{\frac{b}{4}}}$ by Lemma~\ref{lemma:delete_rank}.
Using the pointer from $u$ to the bucket $B_i(u)$ where $i=\floor*{\log_{(1+\lambda)}\max\set*{(1+\lambda)d^+(u),\floor*{\frac{b}{4}}}}$, we may insert $v$ into the correct bucket in $\BigOh(1)$ time. 

    During Insert($\overrightarrow{uv}$) or Delete$(\overrightarrow{uv})$ we loop over at most $\BigOh(\lambda^{-1})$ elements to change their bucket. 
    By Lemma~\ref{lemma:rank_bounded_change}, each element changes their position in the data structure by at most $1$, which can be done in $\BigOh(1)$ time. 
\end{proof}

\subparagraph{Concluding our argument.}
By Lemma~\ref{lemma:rank_invariant_maintenance}, our algorithms maintain Invariant~\ref{inv:degrees_modified} at all times. 
By Lemma~\ref{lem:timebeforerecurse}, our algorithms spend $\BigOh(\lambda^{-1})$ time before recursing (except for the outermost call, which uses $\BigOh(\lambda^{-1}\log n)$ time).
We now make a case distinction.
Either we set $\left( \theta = 1, \, b = 1, \, \eta \in \Theta(\log n) \right)$,
or, we set $\left( \theta = 0, \, b \in \Theta(\log n), \, \eta=3 \right)$.
Because our recursive condition in Algorithms~\ref{alg:insertion_rank}+\ref{alg:deletion_rank} is the same as in Algorithms~\ref{alg:insertion}+\ref{alg:deletion}, we may immediately apply the proofs of Theorem~\ref{thm:invariant1}+\ref{thm:invariant0} to upper bound the recursive depth of our algorithms by $\BigOh(\lambda^{-1} \log \rho)$. 
Thus, the total time for inserting or deleting a single edge in $\overrightarrow{G}^b$ is $\BigOh(\lambda^{-1}\log n+\lambda^{-2}\log\rho)=\BigOh(\lambda^{-2}\log\rho)$
For every update in $G$, we do $\Theta(b)$ updates in $G^b$.
Thus, for both choices of our variables, our algorithms run in time $\BigOh( b \cdot \lambda^{-2} \log \rho)$, and they maintain Invariant~\ref{inv:degrees_modified} for the chosen $\theta \in \{ 0, 1\}$.
Thus, we conclude:

\begin{theorem}
\label{thm:rank0}
Let $G$ be a dynamic graph and $\rho$ be the density of $G$ at update time.
We can choose our variables $\theta = 0$, $\eta=3$, and $b \in \Theta(\log n), b\geq 2$ to maintain an out-orientation $\overrightarrow{G}^b$ in worst case $\BigOh(\log^3 n \log \rho)$  time per operation in $G$, maintaining Invariant~\ref{inv:degrees} for $\overrightarrow{G}^b$ with:
\begin{itemize}[noitemsep, nolistsep]
    \item $\forall v$, the out-degree $d^+(v)$ in $\overrightarrow{G}^b$ is at most $\BigOh( b  \cdot \rho)$, and
    \item $\forall u$, the out-degree of $u$ in $\overrightarrow{G}$ is at most $\BigOh(\rho)$.
\end{itemize}
\end{theorem}

\begin{theorem}
\label{thm:rank1}
Let $G$ be a dynamic graph and $\rho$ be the density of $G$ at time $t$.
We can choose our variables  $\theta = 1$, $b = 1$ and $\eta \in \Theta(\log n)$ to maintain an out-orientation  $\overrightarrow{G}^b = \overrightarrow{G}$ in worst case $\BigOh\paren*{ \log^2 n \log \rho }$ time per update in $G$ such that  Invariant~\ref{inv:degrees_additive} holds for $\overrightarrow{G}$. Moreover:

\begin{itemize}[noitemsep, nolistsep]
 \item $\forall u$, the out-degree $d^+(u)$ in $\overrightarrow{G}$ is at most $\BigOh( \rho + \log n)$, \quad \quad \small{(i.e. $\Delta( \overrightarrow{G}) \in \BigOh( \rho + \log n)$)}
\end{itemize}
\end{theorem}

\section{Improved amortised algorithms}
\label{sec:amortized}

Previously, we relied upon the fact that vertices in $v \in N^-(u)$ were put in buckets based on their exact out-degree $d^+(v)$.
Maintaining these exact values requires $\Omega(\rho)$ update time, and is thus not a suitable option when we aim for polylogarithmic update time.
To this end, we store for all edges $\overrightarrow{uv}$ a single integer $\phi(u, v)$ which we will call their threshold value.
Note that $\phi(u, v) \neq \phi(v, u)$. 
We base our algorithmic logic and analysis on the threshold value instead. 
We maintain for all $u$:

\begin{enumerate}[(a), noitemsep]
    \item The value $d^+(u)$ of the current orientation $\overrightarrow{G}^b$,
    \item The set $N^+(u)$ as a sorted doubly linked list of linked lists $L_j(u)$. 
    Each $L_j(u)$ contains all $w \in N^+(u)$ with $\phi(u, w) = j$ as a linked list in arbitrary order. 
    The linked lists $L_j(u)$ are stored in a linked list sorted by $j$. 
    We maintain a pointer to the location $j = d^+(u)$. 
    \item The set $N^-(u)$ in a sorted doubly linked list of buckets $B_j(u)$.
    Each bucket $B_j(u)$ contains, as a doubly linked list in arbitrary order, all $w \in N^-(u)$ where
    $ 
    j = \floor{ \log_{(1 + \lambda)} d^+(w)}$.
    The vertex $u$ has a pointer to the bucket $B_i(u)$ with $i = \floor*{ \log_{(1 + \lambda)} \max\set*{(1+\lambda)d^+(u),\floor*{\frac{b}{4}}}}$.
\end{enumerate}

\noindent
Any update in $G$, invokes Algorithms~\ref{alg:insertion_in_G}+\ref{alg:deletion_in_G}.
These algorithms now  invoke Algorithm~\ref{alg:insertion_amort} or \ref{alg:deletion_amort} (instead of~\ref{alg:insertion} or \ref{alg:deletion}).
These two in turn invoke the normal add and remove functions (Algorithm~\ref{alg:add}+\ref{alg:remove}). 

\noindent\begin{minipage}[t]{.5 \textwidth}
\null 
 \begin{algorithm2e}[H]
    \caption{\mbox{Insert$(\overrightarrow{uv})$, {\small where $d^+(u)\leq d^+(v)$}}}
    \label{alg:insertion_amort}
    \begin{algorithmic}
      \STATE Add($\overrightarrow{uv}$)
      \WHILE{  $\exists x \in N^+(u)$ with $d^+(u) \geq \max \{  (1 + \lambda) \cdot \phi(u, x), \floor{\frac{b}{4}}$ \} }
        \IF{$d^+(x) + 1 < d^+(u)$ }
        \STATE Remove$(\overrightarrow{ux})$
        \STATE \emph{Insert}$(\overrightarrow{xu})$
     \RETURN
        \ELSE
        \STATE $\phi(u, x) = d^+(u)$
        \STATE Update $x$ in $ N^+(u)$ and $u$ in $N^-(x)$ 
        \ENDIF
      \ENDWHILE
    \end{algorithmic}
  \end{algorithm2e}
\end{minipage}~%
\begin{minipage}[t]{.5\textwidth}
\null
 \begin{algorithm2e}[H]
    \caption{Delete$(\overrightarrow{uv})$}
    \label{alg:deletion_amort}
    \begin{algorithmic}
      \STATE Remove($\overrightarrow{uv}$)
      \FOR{ $i$ in decreasing order} 
      \FORALL{$ x \in B_i(u)$}
      \IF{$d^+(x) > \max \{ (1 + \lambda) d^+(u) + \theta, \floor{\frac{b}{4}} \}$}
        \STATE Add$(\overrightarrow{ux})$
        \STATE \emph{Delete}$(\overrightarrow{xu})$
         \RETURN
        \ELSIF{$\phi(x, u) > (1 + \lambda)d^+(x)$}
        \STATE $\phi(x, u) = d^+(x)$
        \STATE Update $x$ in $N^-(u)$ and $u$ in $N^+(x)$
        \ELSE
        \RETURN
        \ENDIF
        \ENDFOR
      \ENDFOR
    \end{algorithmic}
  \end{algorithm2e}
\end{minipage}

%\noindent\begin{minipage}[t]{.5 \textwidth}
%\null 
% \begin{algorithm}[H]
%    \caption{\mbox{Add$(\overrightarrow{uv})$}}
%    \label{alg:add_amort}
%    \begin{algorithmic}
%%      \IF{ 
 %     $
 %       \nexists \overrightarrow{uv} \in \overrightarrow{G}^b
 %     $
 %    }
 %     \STATE Add $u$ to $N^-(v)$ and $v$ to $N^+(u)$
 %     \ENDIF
%\STATE      
%\STATE Add one edge $\overrightarrow{uv}$ to $\overrightarrow{G}^b$.
%              \STATE $\phi(u, v) = d^+(u)$
%        \STATE Update $v$ in $N^+(u)$ and $u$ in $N^-(v)$ 
%    \end{algorithmic}
%  \end{algorithm}
%\end{minipage}~%
%\begin{minipage}[t]{.5\textwidth}
%\null
% \begin{algorithm}[H]
%    \caption{Remove$(\overrightarrow{uv})$}
%    \label{alg:remove_amort}
%    \begin{algorithmic}
%    \STATE  $d^+(u) = d^+(u) - 1$
%             \STATE Remove one edge $\overrightarrow{uv}$ from $\overrightarrow{G}^b$.
%      \IF{ 
%      $
%        \nexists \overrightarrow{uv} \in \overrightarrow{G}^b
%      $
%%      \STATE Remove $u$ from $N^-(v)$ and $v$ from $N^+(u)$
 %    \ENDIF
 %   \end{algorithmic}
 %%end{minipage}

\begin{lemma}
\label{lem:amor_eq_1}
    Suppose that the graph $G^b$ contains $\overrightarrow{ux}$. Then 
    $d^+(u) \leq  \max\{ (1 + \lambda) \cdot \phi(u, x), \floor{\frac{b}{4}} \}$.
\end{lemma}

\begin{proof}
    Fix some arc $\overrightarrow{ux}$. 
Whenever the value $\phi(u, x)$ is set, it is set to $d^+(u)$. 
Thus, we satisfy the inequality.  The only risk to the desired inequality is increasing $d^+(u)$, whilst it is bigger than $\floor{\frac{b}{4}} $.

Suppose $d^+(u)$  is momentarily increased after adding some arc $\overrightarrow{uv}$. If the inequality for $\overrightarrow{ux}$ is violated, then $x$ is eligible for the while loop.

If the while loop processes $x$, then it will either flip $\overrightarrow{ux}$ or resets $\phi(u,x)$.  If it flips $\overrightarrow{ux}$, then $\overrightarrow{ux}$ is removed from the orientation so there is no inequality to satisfy. Otherwise, we reset $\phi(u,x)$ to $d^+(u)$, which satisfies the inequality.

If the while loop doesn't process $x$, then it must have selected another vertex $y \in N^+(u)$ and flipped $\overrightarrow{uy}$ before processing $x$. In this case, $d^+(u)$ is restored to its previous value when the inequality for $\overrightarrow{ux}$ was satisfied.
\end{proof}

\begin{lemma}
\label{lem:amor_eq_2}
Suppose that $G^b$ contains an edge $\overrightarrow{uz}$. Then:
    $\phi(u, z)  \leq \max \{ (1 + \lambda)^3 ( d^+(z) + \theta), \floor{ \frac{b}{4}} \}$.
\end{lemma}

\begin{proof}
    Fix an arc $\overrightarrow{uz}$. When the value $\phi(u, z)$ is set, it is set to $d^+(u)$ at a point in time where $d^+(u) \leq \max\{  (1 + \lambda) d^+(z)+ \theta, \floor{\frac{b}{4}} \}$. This  satisfies the desired inequality. The only risk to the desired inequality is when $d^+(z)$ decreases.
    Now we perform a case distinction. 
    If $\floor{\frac{b}{4}}  \geq (1 + \lambda)d^+(z) + \theta$, then decreasing $d^+(z)$ did not change the fact that previously $\phi(u, z) \leq \floor{\frac{b}{4}}$.

    Suppose $d^+(z) \geq \frac{b}{4}$, and that it momentarily decreases after deleting $\overrightarrow{zy}$ for some $y$.
    We know that $u \in N^{-}(z)$, so $u \in B_j(z)$ for some integer $j$. 
    In the loop of Delete$(\overrightarrow{zy})$, we consider three cases: 

    (a): we encounter the vertex $u \in N^+(z)$, without hitting any of the two returns. 
    We know that $\phi(u, z) > (1 + \lambda) d^+(u)$ and we set $\phi(u, z)$ to be $d^+(u)$.
    This decreases $\phi(u, z)$, which means that we continue satisfying the inequality.

    (b): we flip an arc $\overrightarrow{tz}$ and return, restoring $d^+(z)$ to its original (inequality-satisfying) value.

    (c): before reaching case (a), we encounter an arc $\overrightarrow{tz}$ for $t \in N^-(z)$ which causes us to hit return (the else in our code). Because we reached this point in the code before case (a) we know that:

    \begin{enumerate}[]
        \item $\phi(u,z) \leq (1 + \lambda) \phi(t,z)$ (we loop over all buckets $B_i(z)$ in decreasing order and did not encounter the vertex $u$. Thus, $u$ is either in the same bucket as $t$ or in a lower bucket). 
        \item $\phi(t,z) < (1 + \lambda) d^+(t)$.
        \item $d^+(t) \leq  \max \{ (1 + \lambda)d^+(z) + \theta, \floor{\frac{b}{4}} \} = (1 + \lambda)d^+(z) + \theta$.
    \end{enumerate}
    Combining these inequalities, we have
    \begin{align*}
    \phi(u,z) \leq (1 + \lambda) \phi(t,z) \leq (1 + \lambda)^2 d^+(t) \leq 
(1 + \lambda)^2 \cdot  ( (1 + \lambda)d^+(z) + \theta)
    \end{align*}
Hence, we recover that $\phi(u,z) \leq (1 + \lambda)^3 (d^+(z) + \theta)$.
\end{proof}

\begin{corollary}
    \label{cor:amortizedinv}
    Our amortised algorithms maintain Invariant~$\theta'$. 
\end{corollary}

\begin{proof}
We combine Lemma~\ref{lem:amor_eq_1} and \ref{lem:amor_eq_2} to get that for all $z \in N^-(u)$: \newline
$d^+(u) \leq \max \{ (1 + \lambda)  \max\{ (1 + \lambda)^3 (d^+(z) + \theta), \floor{\frac{b}{4}} \},\floor{\frac{b}{4}}\}  \leq \max \set*{ (1 + \eta b^{-1}) d^+(z) + 2\theta, \floor{\frac{b}{2}}}$.
\end{proof}

\subsection{Running time analysis}
We now move on to the amortised analysis of the algorithm. At a high level, the idea is as follows. Recall that for each arc $\overrightarrow{ux}$ we have a label $\phi(u,x)$ equal to $d^+(u)$ at some point in time. The labels $\phi(u,x)$ guide the data structure by suggesting arcs to flip. When we operate on an arc $\overrightarrow{ux}$ based on $\phi(u,x)$, if $\overrightarrow{ux}$ is not in fact a good arc to work with, then $d^+(u)$ must have deviated substantially from $\phi(u,x)$, and we reset $\phi(u,x)$ to $d^+(u)$. Loosely speaking, we amortised the effort to relabel  $\overrightarrow{ux}$ against the change to $d^+(u)$.

\begin{lemma}
    Adding an arc to $G^b$ takes $\BigOh( \lambda^{-1})$ amortised time.
\end{lemma}
\begin{proof}
    We note that the recursive depth of Insert($\overrightarrow{uv}$) may be $\BigOh(\rho)$.
    However, we show that the amortised cost of each edge that we process is not too bad. 

    Observe that the net effect of adding an arc, after all flips, is to increase the out-degree $d^+(u)$ of a single vertex $u$ by $1$.

    Now, the running time of adding an arc is proportional to the number of arcs $\overrightarrow{ux}$ processed in the while loop over all recursive calls to Insert($\overrightarrow{uv}$). Each such edge $\overrightarrow{ux}$ has $d^+(u) \geq (1 + \lambda) \phi(u,x) + 2 \theta$, where $\phi(u,x)$ was set to $d^+(u)$ at a previous point in time. Consequently $d^+(u)$ has increased by at least a $(1+\lambda)$-factor since $\phi(u,x)$ was set.

    We amortise the time spent processing arcs $\overrightarrow{ux}$ for fixed $u$ against the increase to $d^+(u)$. Each time an edge insertion results in increase $d^+(u)$, we pay for $\BigOh(1/\lambda)$ units of work distributed uniformly over $N^+(u)$. That is, each $x \in N^+(u)$ receives $\Omega(1 / d^+(u))$ fractional credits. By the time an arc $\overrightarrow{ux}$ is processed in the while loop of insertion, $\overrightarrow{ux}$ has acquired at least one unit of credit, which pays for the time to process it.
\end{proof}

\begin{lemma}
    Removing an arc from $G^b$ takes $\BigOh(\lambda^{-1}\log \rho)$ amortised time.
\end{lemma}
\begin{proof}
    The total running time for a deletion is proportional to the total number of arcs processed in the while loop of Delete$(\overrightarrow{uv})$, over all recursive calls to Delete. Each arc $(x,u)$ processed in the loop (except for the very last one) has one of two outcomes: either it is flipped and we make a recursive call to Delete, or we reset $\phi(x,u)$.

    Since our recursive condition is the same as in Algorithm~\ref{alg:deletion}, we may immediately apply the proofs for upper bounding the recursive depth for deletions from  Theorem~\ref{thm:invariant0} and \ref{thm:invariant1}; showing that the recursive depth is $\BigOh( \lambda^{-1}\log \rho)$. 

    Next we address the number of arcs $\overrightarrow{xu}$ where we reset $\phi(x,u)$.
    We note that $\phi(x,u)$ is only updated when it exceeds $d^+(x)$ by a $(1+\lambda)$-factor.  
    In other words, consider the time $\textsc{start}$ when $\phi(x, u)$ was set.
    At the time $\textsc{end}$ when it is reset to $d^+(x)$, the out-degree of $d^+(x)$ has decreased (by at least a $1 + \lambda$ factor).
    We consider the approximate rank of $d^+(x)$ at two time steps: a lower bound on the rank when $\phi(x, u)$ was set, and an upper bound for when it is reset during a deletion. 
    Formally, we write:  $s = \lfloor \log_{1 + \lambda} \phi(x, u)\rfloor $ and $t = \lceil \log_{1 + \lambda} d^+(x) \rceil$.
    Finally we denote $\delta = s - t$.  
    Note that to update our data structure on $N^-(u)$, we need to move $\phi(x, u)$ by $\Theta(\delta)$ buckets.      
    We make a case distinction based on whether $\delta < 3$ or $\delta \geq 3$. 

    \textbf{Case 1: $\delta < 3$.}
    In this case when setting $\phi(x, u) = d^+(x)$ we need to move $x$ $\BigOh(1)$ buckets in the data structure on $N^-(u)$.
    By the time $\phi(u,x)$ is reset, $d^+(x)$ has decreased by at least $(1 + \lambda)^{s} - (1+ \lambda)^{s - 1} \geq \frac{\lambda \phi(u, x)}{(1 + \lambda)}$ since $\phi(u,x)$ was set. The net effect of each deletion (after all flips and recursive calls) is to decrease the degree of a single vertex $x$ by $1$. When this occurs, we pay for $4 / \lambda$ units of work that are distributed uniformly over $N^+(x)$. Consequently, by the time we reset $\phi(u,x)$ in a call to Delete($\overrightarrow{uv}$) (for some $v$), $x$ has already acquired one fractional unit of work to pay for the $\BigOh(1)$ work.

    \textbf{Case 2: $\delta \geq 3$.}
    In this case, the rank of $d^+(x)$ decreased by at least $\Theta(\delta)$ and at least three levels. 
The net effect of each deletion (after all flips and recursive calls) is to decrease the degree of a single vertex $x$ by $1$, at which point we distribute $4 / \lambda$ credits over $N^-(x)$. 
    Between time $\textsc{start}$ and $\textsc{end}$, the out-degree $d^+(x)$ may arbitrarily increase and decrease. However, we can always find a sequence of (not necessarily consecutive) edge deletions $S = \{ (\alpha, \beta)_i \}$ such that after deletion $(\alpha, \beta)_i$ in $G^b$, the out-degree $d^+(x)$ decremented by one, and for any pair of consecutive edge deletions $(\alpha, \beta)_i$ $(\alpha', \beta')_{i+1}$ in $S$, the out-degree of $x$ at the end of deleting $(\alpha, \beta)_i$  equals the out-degree of $x$ at the start of deleting $(\alpha', \beta')_{i+1}$.
    Denote by $S^i \subset S$ all deletions in $S$ where after the deletion, the vertex $x$ has a rank $t + i + 1$ for $0 < i < s - t$. 
    For every deletion in $S^i$, we distribute $4 / \lambda$ of units of work over $N^+(x)$.

    For all $i \in (0, s - t - 1)$, decreasing the rank of $x$ from $t + i + 1$ to $t + i$ requires exactly $(1+\lambda)^{t+i + 1} - (1+\lambda)^{t+i} = \lambda (1 + \lambda)^{t+i}$ deletions.
    Thus, $S^i$ has exactly $(1+\lambda)^{t+i + 1} - (1+\lambda)^{t+i} = \lambda (1 + \lambda)^{t+i}$. Per definition of $S^i$, after each deletion, $N^+(x)$ has at most $(1 + \lambda)^{t + i + 2}$ out-edges. 
    Thus, whenever we distribute after each deletion $4 / \lambda$ credits over all $N^+(x)$, the number of credits per edge $C$ is at least:
    \[
    C \geq \sum_{i = 0}^{\delta - 2}   4 \lambda^{-1} \frac{\textnormal{\# of deletions in } S^i }{\textnormal{out-degree of } x \textnormal{ during deletions in } S^i }  =  \sum_{i = 0}^{\delta - 2}  4 \lambda^{-1}  \frac{ \lambda (1 + \lambda)^{t + i} }{ (1 + \lambda)^{t+i + 2}  } = \sum_{i = 0}^{\delta - 2} \frac{4 \lambda}{\lambda (1 + \lambda)^2} \geq  \delta - 2 
    \]
    Here, the second-to-last inequality follows from the fact that $\lambda \leq 1$ and thus $\frac{4}{(1 + \lambda)^2} \geq 1$.  
    Hence, for $\delta \geq 3$ we have acquired $O(\delta)$ credits on every edge in $N^+(x)$, which we may  use to pay for relocating $x$ by $\delta$-buckets.

   % \textbf{Case 2: $\delta \geq 3$.}
   % In this case when setting $\phi(x, u) = d^+(x)$ we need to move $x$ $\BigOh(\delta)$ buckets in the data structure on $N^-(u)$.
   % By the time $\phi(u,x)$ is reset, $d^+(x)$ has decreased at least.
   % Consider between then and now, any sequence of deletions of edges $(x, w)$ such that the out-degree of $d^+(x)$ after each deletion is a linear sequence. 
   % In this sequence, there are at least this many deletions:
   % \[
   % (1 + \lambda)^s - (1 + \lambda)^{s - \delta}  = ( (1 + \lambda)^\delta - 1  ) (1 + \lambda)^{s - \delta} = ( (1 + \lambda)^\delta - 1  ) \cdot \frac{Z}{(1 + \lambda)^\delta}
   % \]
   % Imagine distributing $2$ credits of work after each deletion in the sequence over all out-neighbors of $x$, for $N$ deletions.
   % Then the total number of credits per remaining out-neighbour of $x$ is at least: 
   % \[
   % C(s, N) = \sum_{i = 1}^N \frac{2}{Z - i} = 2 \cdot 2 ( H_{-(1+\lambda)^s} - H_{N -(1 + \lambda)^s} )
   % \]
   % We substitute $N =  ( (1 + \lambda)^\delta - 1  ) (1 + \lambda)^{s - \delta}$, and we want that each vertex has at least $\frac{\delta}{2}$ %credits. 
   % Thus we get:
   % \[
   % \sum_{i=1}^{((1 + \lambda)^\delta - 1) (1 + \lambda)^{s - \delta}} \frac{2}{ (1 + \lambda)^s - i }>= \frac{\delta}{2} \Rightarrow 4 H_{-(1 + \lambda)^s} \geq  4 H_{-(1 + \lambda)^(s - \delta)} + \delta
  %  \]

\end{proof}

%\begin{proof}
%    ...Finally we denote $\delta = s - t$.
%    To update our data structure on $N^-(u)$, we need to move $\phi(x,u)$ by $\Theta(\delta)$ buckets. 
%    The net effect of each deletion is to decrease the degree of a single vertex $x$ by $1$. When this %occurs, we pay for $4/\lambda$ units of work that are distributed uniformly over $N^+(x)$.

  %  We want to argue that by the time $\Phi(u,v)$ is reset, it has acquired $\Omega(\delta)$ credits. These credits pay for the $O(\delta)$ work needed to move $\phi(x,u)$ by $O(\delta)$ buckets.

  %  The number of (fractional) credits acquired by $\phi(x,u)$ is bounded below by $\Omega(\log_{1 + \lambda} (\phi(x,u) / d^+(x))$.
 %   We have 
  %  \begin{align*}
  %      \log_{1 + \lambda} \frac{\phi(x,u)}{d^+(x)}
   %     \geq
  %      \lfloor \log_{1 + \lambda} \phi(x,u) \rfloor - \lceil \log_{1 + \lambda} \phi(d^+(x)) \rceil
   %     =          s - t
    %    = \delta.
  %  \end{align*}
  % Thus $\phi(x,u)$ has at least $\Omega(\delta)$ credits by the time it is reset, as desired.    
%\end{proof}

We may now apply Corollary~\ref{cor:inv0} and Corollary~\ref{cor:inv1}.
These set $\eta$ and $b$ such that $\lambda^{-1} \in \BigOh(\log n)$. 
We note that for every insertion in $G$, we insert $b$ edges in $G^b$.
Thus, we conclude:

\begin{theorem}
\label{thm:amor0}
Let $G$ be a dynamic graph and $\rho$ be the density of $G$ at update time.
We can choose our variables $\theta = 0$, $\eta=3$, and $b \in \Theta(\log n), b\geq 2$ to maintain an out-orientation $\overrightarrow{G}^b$ in $\BigOh(\log^2 n \log \rho)$ amortized time per operation in the original graph $G$, maintaining Invariant~\ref{inv:degrees} for $\overrightarrow{G}^b$ with:
\begin{itemize}[noitemsep, nolistsep]
 \item $\forall u$, the out-degree $d^+(u)$ in $\overrightarrow{G}$ is at most $\BigOh( \rho + \log n)$, \quad \quad \small{(i.e. $\Delta( \overrightarrow{G}) \in \BigOh( \rho + \log n)$)}
\end{itemize}
\end{theorem}

\begin{theorem}
\label{thm:amor1}
Let $G$ be a dynamic graph and $\rho$ be the density of $G$ at time $t$.
We can choose our variables $\theta = 1$, $b = 1$ and $\eta \in \Theta(\log n)$  to maintain an out-orientation  $\overrightarrow{G}^b = \overrightarrow{G}$ in amortized $\BigOh\paren*{ \log n \log \rho }$ time per update in $G$ such that  Invariant~\ref{inv:degrees_additive} holds for $\overrightarrow{G}$. Moreover:
\begin{itemize}[noitemsep, nolistsep]
    \item $\forall v$, the out-degree $d^+(v)$ in $\overrightarrow{G}^b$ is at most $\BigOh( b  \cdot \rho)$, and
    \item $\forall u$, the out-degree of $u$ in $\overrightarrow{G}$ is at most $\BigOh(\rho)$.
\end{itemize}
\end{theorem}

\section{Obtaining  \texorpdfstring{$(1 + \eps)$}{(1+ε)} Approximations }
\label{sec:onepluseps}
Finally, we note that we can choose our variables carefully to obtain a $(1 + \eps)$ approximations  of the maximum subgraph density or minimum out-degree. 
 Theorem~\ref{thm:rank0} implies that, for suitable choices of $\eta$ and $b$, we can for any graph $G$ maintain a directed graph $\overrightarrow{G}^b$ (where $G^b$ is the graph $G$ with every edge duplicated $b$ times) such that $\overrightarrow{G}^b$  maintains Invariant~\ref{inv:degrees}.
 By Theorem~\ref{thm:structural}, $\overrightarrow{G}^b$ approximates the densest subgraph of $G$ and the minimum out-orientation of $G^b$ (where the approximation factor is dependent on $\beta$ and $\eta$). 
  The running time of the algorithm is $\BigOh( b^3 \cdot \log \alpha)$ where $\alpha$ is the arboricity of the graph. 
  In this section we show that for any $0 < \eps < 1$, we can choose  an $\eta>0$ and a $b \in \BigOh(\eps^{-2} \log n)$ to ensure that $\overrightarrow{G}^b$ maintains a:
  
  \begin{itemize}
      \item $(1 + \eps)$-approx. of \textbf{the maximum densest subgraph} of $G$ in $\BigOh(\eps^{-6} \log^3 n \log \alpha)$ time.
      \item $(1 + \eps)$-approx. of the minimum out-orientation of $G^b$. This implies an explicit \textbf{$(2 + \eps)$-approximation of the minimum out-orientation} of $G$  in $\BigOh(\eps^{-6} \log^3 n \log \alpha)$ time.
      \item  $(1 + \eps)$-approx. of the minimum out-orientation of $G^b$. 
      Through applying clever rounding introduced by Christiansen and Rotenberg~\cite{christiansenICALP} we obtain an explicit \textbf{$(1 + \eps)$-approximation of the minimum out-orientation} of $G$. 
   By slightly opening their black-box algorithm, we can show that applying their technique does not increase our running time.  Thus, our total running time is thus $\BigOh(\eps^{-6} \log^3 n \log \alpha)$.
  \end{itemize}

\subsection*{Obtaining a \texorpdfstring{$(1 + \eps)$}{(1+ε)} Approximation for Densest Subgraph}

\begin{corollary}\label{cor:eps_subgraphdensity}
Let $G$ be a dynamic graph subject to edge insertions and deletions with  adaptive maximum subgraph density $\rho$. Let $G^b$ be $G$ where every edge is duplicated $b$ times. Let $0 \leq \epsilon < 1$. We can maintain an orientation $\overrightarrow{G}^b$ such that 
\[
\rho \leq  b^{-1} \cdot \Delta( \overrightarrow{G}^b) \leq (1+\eps)\rho
\]
with update time $\BigOh(\eps^{-6}\log^3(n)\log \rho )$ per operation in $G$.
\end{corollary}
\begin{proof}
We apply Theorem~\ref{thm:rank0} in order to maintain an out-orientation satisfying Invariant~\ref{inv:degrees}, which by Theorem~\ref{thm:structural} satisfies $\rho(G^b) \leq \Delta(\overrightarrow{G}^b) \leq (1+\gamma)(1+\eta\cdot{}b^{-1})^{k_{\max}} \rho(G^b)$. 
By setting $\gamma=\frac{\eps}{2}$,  $\eta = 3$, $b = \ceil{\gamma^{-1}\eta\log_{(1+\gamma)}n}\in\BigOh(\eps^{-2}\eta \log{n})$, we satisfy the conditions of the Theorem. Since $k_{\max} \leq \log_{1+\gamma} n$, we find that
\[
(1+\eta\cdot{}b^{-1})^{k_{\max}} \leq e^{\eta b^{-1} \cdot k_{\max}} \leq e^{\gamma} \leq 1+2\gamma = 1+\eps
\]
where the last inequality comes from the fact that for $0 \leq x \leq 1$, we have $e^x \leq 1+2x$.
\end{proof}

\begin{observation}
The algorithm of Corollary~\ref{cor:eps_subgraphdensity} can in $\BigOh(1)$ time per operation, maintain the integers: $b^{-1}$, $\Delta(\overrightarrow{G}^b)$ and thus a $(1 + \eps)$ approximation of the value of the density of $G$. 
\end{observation}

\noindent
However, to actually output any such realizing subgraph, a bit more of a data structure is needed:

\begin{lemma} \label{lemma:SDE}
For a fully-dynamic graph $G$, there is an algorithm 
that explicitly maintains a $(1+\eps)$ approximation of the maximum subgraph density in 
%$\BigOh(\eps^{-1} \log^2 n)$ additional 
$\BigOh(\eps^{-6}\log ^3 n \log \alpha )$ total
time per operation, and that
can output a subgraph realizing this density in  $\BigOh(\texttt{occ})$ time where $\texttt{occ}$ is the size of the output.
% and the size of the densest subgraph in $\BigOh(1)$ time.
\end{lemma}

%\begin{lemma} \label{lemma:SDE}
%Consider an algorithm that maintains for a dynamic graph $G$, a directed graph $\overrightarrow{G}^b$ for which Invariant~$0$ holds (with $\eta > 1280$, $b = \BigOh(\eps'^{-2}\eta \log{n})$ and $\gamma = \eps'$). Moreover, let for every operation in $G$, this algorithm change the degree of $r_o$ vertices in $\overrightarrow{G}^b$ in $t$ total time. 
%Then  we can dynamically maintain a data structure in $\BigOh(t + r_o \log n)$ time that can report a $(1+\eps)$ approximation of the densest subgraph of $G$ in $\BigOh( \eps^{-1} \log^2 n + c)$ time where $c$ is the size of the output.
%\end{lemma}

\begin{proof}
We use Corollary~\ref{cor:eps_subgraphdensity} to dynamically maintain an orientation $\overrightarrow{G}^b$ in $\BigOh(\eps^{-6}\log^3(n)\log \rho )$ per operation in $G$.
Recall (Theorem~\ref{thm:structural}) that we defined for non-negative integers $i$ the sets:
\[
T_i := \set*{ v\in V \suchthat d^+(v) \geq \Delta\paren*{\overrightarrow{G}^b} \cdot \paren*{1 + \eta\cdot b^{-1}}^{-i} } 
\]

(note that since we maintain Invariant~\ref{inv:degrees}, the constant $c$ in the previous definition is zero).

Let $k$ be the smallest integer such that $|T_{k+1}| < (1 + \gamma) |T_k|$). Moreover, we showed in Corollary~\ref{cor:eps_subgraphdensity} that $k$ is upper bounded by $\BigOh( \eps^{-1} \log n)$.  
We show in Section~\ref{sec:struc} that (the induced subgraph of the vertex set) $T_{k+1}$ is an approximation of the densest subgraph of $\overrightarrow{G}^b$ (and therefore of $G$). 
We store the vertices of $\overrightarrow{G}^b$ as leaves in a balanced binary tree, sorted on their out-degree. 
Since every change in $G$, changes at most $\BigOh(b \log n \log \rho) = \BigOh(\eps^{-2} \log^2 n \log \rho)$ out-degrees in $\overrightarrow{G}$, we can maintain this binary tree in $\BigOh(\eps^{-2} \log^3 n \log \rho)$ additional time per operation in $G$. 

Each internal node of the balanced binary tree stores the size of the subtree rooted at that node. 
Moreover, we store the maximum out-degree $\Delta(\overrightarrow{G^b})$ as a separate integer, and a doubly linked list amongst the leaves.

After each operation in $G$, for each integer $i \in [0, \eps^{-1} \log n]$, we determine how many elements there are in $T_i$ as follows: 
first, we compute the value $V_i = \Delta(\overrightarrow{G}^b) \cdot (1 + \eta\cdot b^{-1})^{-i}$.
%$ - c\sum_{j = 1}^i (1 + \eta\cdot b^{-1})^{-j}$ in $\BigOh(1)$ time (using t
Then, we identify in $\BigOh(\log n)$ time how many vertices have out-degree at least $V_i$ (thus, we determine the size of $T_i$). 
It follows that we identify $T_k$ in $\BigOh(\eps^{-1} \log^2 n)$ additional time. We store a pointer to the first leaf that is in $T_k$. 
If we subsequently want to output the densest subgraph of $G$, we traverse the $\texttt{occ}$ elements of $T_k$ in $\BigOh(\texttt{occ})$ total time by traversing the doubly linked list of our leaves.
\end{proof}

\subparagraph{Related Work}
While results for densest subgraph \cite{BahmaniKV12,BhattacharyaHNT15,EpastoLS15} can be used to estimate maximum degree of the best possible out-orientation, it is also interesting in its own right.
Sawlani and Wang~\cite{sawlani2020near} maintain a $(1 - \eps)$-approximate densest subgraph in worst-case time $\BigOh(\eps^{-6}\log^4 n )$ per update where they maintain an \emph{implicit} representation of the approximately-densest subgraph. They write that they can, in $\BigOh(\log n)$ time, identify the subset $S \subseteq V$ where $G[S]$ is the approximately-densest subgraph and they can report it in $\BigOh(|S|)$

\subsection*{Obtaining an almost $(1 + \eps)$ Approximation for Minimum Out-orientation }

By Corollary~\ref{cor:eps_subgraphdensity}, we can dynamically maintain for every graph $G$,  a directed graph $\overrightarrow{G}^b$ (where each edge in $G$ is duplicated $b$ times) such that the maximum out-degree in $\overrightarrow{G}^b$ is at most a factor $(1 + \eps)$ larger than the minimum out-orientation of $G^b$. 
For every edge $(u, v)$ in $G$, we can now store a counter indicating how many edges point (in $G^b$) from $u$ to $v$, or the other way around. 
The naive rounding scheme, states that the edge $(u, v)$ is directed as $\overrightarrow{uv}$ whenever there are more edges directed from $u$ to $v$. For any edge, we can decide its rounding in $\BigOh(1)$ time, thus we conclude:

\begin{observation}\label{obs:rounding}
We can maintain for a graph $G$ an orientation $\overrightarrow{G}$ where each vertex has an out-degree of at most 
$(2+\varepsilon)\alpha$ 
 with update time $\BigOh(\eps^{-6}\log^3(n)\log \rho )$ per operation.
\end{observation}

\noindent
Obtaining a $(1 + \eps)$-approximation of the minimum out-orientation of $G$ is somewhat more work. 
Christiansen and Rotenberg~\cite{christiansenICALP} show how to dynamically maintain an explicit out-orientation on $G$ of at most $(1 + \eps) \alpha + 2$ out-edges. In their proofs, Christiansen and Rotenberg~\cite{christiansenICALP} rely upon the algorithm by Kopelowitz, Krauthgamer, Porat and Solomon~\cite{KopelowitzKPS13}. 
By replacing the KKPS~\cite{KopelowitzKPS13} algorithm by ours in a black-box like manner, we obtain the following:

\begin{theorem}
\label{thm:epsapprox}
Let $G$ be a dynamic graph subject to edge insertions and deletions. We can maintain an orientation $\overrightarrow{G}$ where each vertex has an out-degree of at most $(1+\varepsilon)\alpha + 2$ 
 with update time $\BigOh(\eps^{-6} \log^3 n \log \alpha)$ per operation in $G$, where $\alpha$ is the arboricity at the time of the update.
\end{theorem}

\noindent
The proof follows immediately from the proof Theorem 26 by Christiansen and Rotenberg~\cite{christiansenICALP} (using Corollary~\ref{cor:eps_subgraphdensity} as opposed to~\cite{KopelowitzKPS13}). 
For the reader's convenience, we will briefly elaborate on how this result is obtained and how we can apply Corollary~\ref{cor:eps_subgraphdensity}.
For the full technical details, we refer to the proof of Theorem 26 in~\cite{christiansenICALP}.

\begin{enumerate}
    \item 
Christiansen and Rotenberg consider a graph $G$ with arboricity $\alpha$.
Moreover, they construct a directed graph $\overrightarrow{G}^b$ which is the graph $G$ where every edge in $G$ is duplicated $b \in \BigOh(\eps^{-2} \log n)$ times.\footnote{In \cite{christiansenICALP}, Christiansen and Rotenberg choose the duplication constant to be $\gamma$ and write $G^\gamma$.} 
Every operation in $G$ triggers $\BigOh(b)$ operations in $\overrightarrow{G}^b$.

\item On the graph $G^b$, they run the algorithm by~\cite{KopelowitzKPS13} to maintain an orientation of $\overrightarrow{G}^b$ where each vertex has an out-degree of at most $\Delta(\overrightarrow{G}^b) = (1 + \eps')\alpha \cdot b + \log_{(1 + \eps')} n$ for $\eps' = \theta(\eps)$ (for instance $\eps' = \eps/4$ works). 
The KKPS~\cite{KopelowitzKPS13} algorithm uses per operation in $G^b$:\footnote{Christiansen and Rotenberg deliberately use the adaptive variant of KKPS~\cite{KopelowitzKPS13}.}
\begin{itemize}[noitemsep, nolistsep]
    \item $\BigOh\paren*{ \paren*{\Delta(\overrightarrow{G}^b) }^2 } = \BigOh\paren*{ (1 + \eps)^2 \alpha^2 b^2 + \eps^{-4} \log^2 n }  = \BigOh( \eps^{-4} \alpha^2 \log^2 n)$ time, and
    \item $\BigOh\paren*{ \Delta(\overrightarrow{G}^b) } = \BigOh\paren*{ (1 + \eps)\alpha b + \eps^{-2} \log n }  = \BigOh( \eps^{-2} \alpha \log n)$ \emph{combinatorial changes } in $\overrightarrow{G}^b$.
    (here, a combinatorial change either adds, removes, or flips an edge in $\overrightarrow{G}^b$).
\end{itemize}

\item Finally, they deploy a clever rounding scheme to transform the orientation $\overrightarrow{G}^b$ into an orientation of $G$ where the out-degree of each vertex in $\overrightarrow{G}$ is at most a factor $\frac{1}{b}$ the out-orientation of $\overrightarrow{G}^b$, plus two.
Thus, they ensure that each vertex has an out-degree of at most: 
\[
(1 + \eps') \alpha + b^{-1} \log_{1 + \eps'} n  + 2 \leq (1 + \eps') \alpha + \frac{\eps'^2}{\log n} \cdot \frac{2\log n}{\eps'} + 2 = (1 + \eps) \alpha + 2
\]
since $\alpha \geq 1$ if the graph has at least one edge (otherwise the claim is vacant).
They achieve this in $\BigOh(\log n)$ additional time per combinatorial change in $\overrightarrow{G}^b$.
Specifically:
\begin{itemize}
    \item They consider for every edge $(u, v)$ in $G$ its partial orientation (i.e. how many edges in $G^b$ point from $u$ to $v$ or vice versa). 
If the partial orientation contains sufficiently many edges directed from $u$ to $v$, the edge in $G$ gets rounded (directed from $u$ to $v$). 
\item Let $H$ be a (not necessarily maximal) set of  edges in $G$ whose direction can be determined in this fashion. They call $H$ a \emph{refinement}. Christiansen and Rotenberg choose $H$ such that in the rounded, directed graph $G-H$ each vertex has an out-degree of at most $(1 + \eps) \alpha$.
\item 
Christiansen and Rotenberg show that $H$ always can be made into a forest.
For all edges in $H$, they no longer explicitly store the $b$ copies in $\overrightarrow{G}^b$. Instead, they store for edges in $H$ their (partial) orientation as an integer in $[0, b]$. 
The forest $H$ gets stored in a top tree where each interior node stores the minimum and maximum partial orientation of all its children.
For any path or cycle in $H$, they can increment or decrement all orientation integers by 1 in $\BigOh(\log n)$ time by lazily updating these maxima and minima in the top tree. 
For each edge in $H$, one can obtain the exact partial orientation in $\BigOh(\log n)$ additional time by adding all lazy updates in the root-to-leaf path of the top tree.
\item 
In addition, they show how to dynamically maintain a $2$-orientation on the forest $H$ in $\BigOh(\log n)$ update time per insertion in the forest. Adding the directed edges from the forest to $G$ ensures that each vertex has an out-degree of at most $(1 + \eps) \alpha + 2$.
\item For each combinatorial change in $\overrightarrow{G}^b$, they spend $\BigOh(\log n)$ time.
Specifically:
\begin{itemize}
    \item each combinatorial change in $\overrightarrow{G}^b$ may remove an edge from the forest. The edge can be rounded in $\BigOh(1)$ time and removed from the top tree in $\BigOh(\log n)$ time. 
    \item each combinatorial change may force an edge in $G$ into the refinement and thus possibly creating a cycle.
    \item When creating a cycle, the authors augment the cycle such that at least one edge on the cycle may be expelled from the refinement. They (implicitly) increment or decrement all orientation integers along the cycle using the lazy top tree in $\BigOh(\log n)$ total time. 
    \item Augmenting a cycle causes the out-degree to remain the same for all elements on the cycle. Hence, the Invariants of KKPS~\cite{KopelowitzKPS13} (and our Invariant~\ref{inv:degrees}) stay unchanged and the augmentation does not trigger any further operations in $\overrightarrow{G}^b$. Note also that they specifically always leave at least one duplicate edge in each direction, so that no additional data structures need be updated. 
\end{itemize}
\item The final edge along the augmented path may subsequently be rounded and added to $G-H$. Thus, spending $\BigOh(\log n)$ time per combinatorial change in $\overrightarrow{G}^b$.
\end{itemize}
\end{enumerate}

\noindent
It follows through these three steps that the algorithm in \cite{christiansenICALP} has a running time of: 
\[
\BigOh\paren*{ b \cdot \paren*{ \paren*{ \Delta(\overrightarrow{G}^b) }^2 +  \Delta(\overrightarrow{G}^b) \log n } } = \BigOh\paren*{ \eps^{-6} \alpha^2 \log^3 n }
\]

\noindent
Given the results in this paper, we can instead apply our results as follows: 

\begin{enumerate}
    \item We again choose $b \in \BigOh(\eps^{-2} \log n)$. Each operation in $G$ triggers $\BigOh(b)$ operations in $\overrightarrow{G}^b$. 
    \item We apply Theorem~\ref{thm:rank0} (or conversely Corollary~\ref{cor:eps_subgraphdensity}) to maintain $\overrightarrow{G}^b$ such that each vertex has an out-degree of at most $\Delta(\overrightarrow{G}^b) = (1 + \theta(\eps)) \alpha b$. 
    We proved that this algorithm takes: 
    \begin{itemize}
        \item  $\BigOh(b^2 \log \alpha)$ time per operation in $\overrightarrow{G}^b$, but
        \item only triggers $\BigOh(b \log \alpha)$ combinatorial changes (edge flips) in $\overrightarrow{G}^b$.
    \end{itemize}
    \item Finally, we apply the rounding scheme by Christiansen and Rotenberg which requires $\BigOh(\log n)$ time per \emph{combinatorial change} in $\overrightarrow{G}^b$.
\end{enumerate}

\noindent
Our total running time is (our algorithm + rounding scheme per combinatorial change):
\[
\BigOh\paren*{b \cdot b \cdot b \log \alpha  + b \cdot b \log \alpha \cdot \log n } = \BigOh\paren*{\eps^{-6} \log^3 n \log \alpha}.
\]

\subparagraph{Related Work}
Historically, four criteria are considered when designing dynamic out-orientation algorithms: the maximum out-degree, the update time (or the recourse), amortised versus worst-case updates, and the adaptability of the algorithm to the current arboricity. 

Brodal and Fagerberg~\cite{Brodal99dynamicrepresentations} were the first to consider the out-orientation problem in a dynamic setting. 
They showed how to maintain an $\BigOh(\alpha_{\max})$ out-orientation with an amortised update time of $\BigOh(\alpha_{\max}+ \log{n})$, where $\alpha_{\max}$ is the maximum arboricity throughout the entire update sequence.
Thus, their result is adaptive to the current arboricity as long as it only increases. 
He, Tang, and Zeh~\cite{HeTZ14} and Kowalik~\cite{10.1016/j.ipl.2006.12.006} provided different analyses of Brodal and Fagerbergs algorithm resulting in faster update times at the cost of worse bounds on the maximum out-degree of the orientations. 
Henzinger, Neumann, and Wiese~\cite{henzinger2020explicit} gave an algorithm able to adapt to the current arboricity of the graph, achieving an out-degree of $\BigOh(\alpha)$ and an amortised update time \emph{independent} of $\alpha$, namely $\BigOh(\log^2 n)$.
Kopelowitz, Krauthgamer, Porat, and Solomon~\cite{KopelowitzKPS13} showed how to maintain an $\BigOh(\alpha+\log n)$ out-orientation with a worst-case update time of $\BigOh(\alpha^2 + \log^2 n)$ fully adaptive to the arboricity. 
Christiansen and Rotenberg~\cite{christiansenICALP,christiansenMFCS} lowered the maximum out-degree to $(1+\varepsilon)\alpha+2$ incurring a worse update time of $\BigOh(\varepsilon^{-6}\alpha^2\log^3 n)$.
Finally, Brodal and Berglin~\cite{berglinetal:LIPIcs:2017:8263} gave an algorithm with a different trade-off; they show how to maintain an $\BigOh(\alpha_{\max}+\log n)$ out-orientation with a worst-case update time of $\BigOh(\log n)$. This update time is faster and independent of $\alpha$, however the maximum out-degree does not adapt to the current value of $\alpha$.

\section{Applications}
\label{app:applications}
In this section, we show how to combine our two trade-offs for out-orientations (theorems \ref{thm:rank0}, \ref{thm:rank1} with existing or folklore reductions, obtaining improved algorithms for maximal matching, arboricity decomposition, and matrix-vector product.

\subsection{Maximal matchings}
  
  For our application in maximal matchings, we first revisit the following result.
  The authors have not seen this theorem stated in this exact generality in the literature, but similar statements appear in~\cite{PelegS16}, \cite{NeimanS16}, and \cite{berglinetal:LIPIcs:2017:8263}
  
\begin{lemma}[Folklore] \label{thm:folklore}
Suppose one can maintain an edge-orientation of a dynamic graph, 
that has $t_u$ update time, 
that for each update performs at most $r_u$ edge re-orientations (direction changes), and that maintains a maximum out-degree of $\le n_o$. Then there is a dynamic maximal matching algorithm\footnote{When the update time $t_u$ is worst-case, the number of re-orientations $r_u$ is upper bounded by $t_u$.} whose update time is $\BigOh(t_u+r_u+n_o)$.
\end{lemma}

\begin{proof}[Proof of Lemma~\ref{thm:folklore}]
Each vertex maintains two doubly-linked lists over its in-neighbors (one for the matched, and one for the available in-neighbors) called \emph{in-lists}
and a doubly-linked list of its out-neighbors called the \emph{out-list}. When a vertex becomes available because of an edge deletion, it may match with the first available in-vertex if one exists.
If no such in-vertex exists, it may propose a matching to its $\le n_o$ out-neighbors in the out-list, and then match with an arbitrary one of these if any is available. When a vertex $v$ changes status between matched and available, it notifies all vertices in its out-list, who move $v$ between in-lists in $\BigOh(1)$ time. Finally, when an edge changes direction, each endpoint needs to move the other endpoint between in- and out-lists. 

The bookkeeping of moving vertices between unordered lists takes constant time.
For each edge insertion or deletion, we may spend additionally $\BigOh(n_o)$ time  proposing to or notifying to out-neighbors to a vertex, for at most two vertices for each deletion or insertion respectively.
\end{proof}

\noindent
With this application in mind, some desirable features of out-orientation algorithms become evident: 
\begin{itemize}[noitemsep, nolistsep]
    \item we want the number of out-edges $\mout$ to be (asymptotically) low, and
    \item we want the update time to be efficient, preferably deterministic and worst-case.
\end{itemize}

\noindent
Here, a parameter for having the number of out-edges asymptotically as low as possible, can be sparseness measures such as the maximum subgraph density or the arboricity of the graph. 
%(see Section~\ref{sec:tech} for a formal definition)
An interesting challenge for dynamic graphs is that the density may vary through the course of dynamic updates, and we prefer not to have the update time in our current sparse graph to be affected by a brief occurrence of density in the past. 
In the work of Henzinger, Neumann, and Wiese, they show how it is possible to adjust to the current graph sparseness in the amortised setting~\cite{henzinger2020explicit}. In this paper, however, we are interested in the case where both the update time is worst-case and the number of re-orientations is bounded. 
One previous approach to this challenge is to take a fixed upper bound on the sparseness as parameter to the algorithm, and then use $\log n$ data structures in parallel~\cite{sawlani2020near}. Since we want the number of re-orientations to also be bounded, we cannot simply change between two possibly very different out-orientations that result from different bounds on the sparseness. Any scheme for deamortising the switch between structures would be less simple than the approach we see in this paper.

\begin{corollary}\label{cor:match}
    There is a deterministic dynamic maximal matching algorithm with worst-case $\BigOh(\alpha + \log ^2 n \log \alpha)$ update time, where $\alpha$ is the current arboricity of the dynamic graph. The algorithm also implies a $2$-approximate vertex cover in the same update time. 
\end{corollary}

\subparagraph{Related Work}
Matchings have been widely studied in dynamic graph models. Under various plausible conjectures, we know that a maximum matching cannot be maintained even in the incremental setting and even for low arboricity graphs (such as planar graphs) substantially faster than $\Omega(n)$ update time \cite{AbboudD16,AbboudW14,HenzingerKNS15,KopelowitzPP16,Dahlgaard16}.
Given this, we typically relax the requirement from maximum matching to maintaining matchings with other interesting properties. 
%One such relaxation is to consider approximate matchings, which we discuss in the appendix.
One such relaxation is to require that the maintained matching is only \emph{maximal}. The ability to retain a maximal matching is frequently used by other algorithms, notably it immediately implies a $2$-approximate vertex cover. 
In incremental graphs, maintaining a maximal matching is trivially done with the aforementioned greedy algorithm. 
For decremental\footnote{Maintaining an \emph{approximate maximum matching} decrementally is substantially easier than doing so for fully dynamic graphs. Indeed, recently work by \cite{AssadiBD22} matches the running times for approximate maximum matching in incremental graphs \cite{GLSSS19}. However, for maximal matching, we are unaware of work on decremental graphs that improves over fully dynamic results.} or fully dynamic graphs, there exist a number of trade-offs (depending on whether the algorithm is randomised or determinstic, and whether the update time is worst case or amortised).  Baswana, Gupta, and Sen~\cite{BaswanaGS15} and
Solomon~\cite{Solomon16} gave randomised algorithms maintaining a maximal matching with $\BigOh(\log n)$ and $\BigOh(1)$ amortised update time. These results were subsequently deamortised by Bernstein, Forster, and Henzinger \cite{BernsteinFH21} with only a $\polylog n$ increase in the update time. For deterministic algorithms, maintaining a maximal matching is substantially more difficult.
Ivkovic and Lloyd~\cite{IvkovicL93} gave a deterministic algorithm with $\BigOh((n+m)^{\sqrt{2}/2})$ worst case update time. This was subsequently improved to $\BigOh(\sqrt{m})$ worst case update time by Neiman and Solomon \cite{NeimanS16}, which remains the fastest deterministic algorithm for general graphs.

Nevertheless, there exist a number of results improving this result for low-arboricity graphs. Neiman and Solomon \cite{NeimanS16} gave a deterministic algorithm that, assuming that the arboriticty of the graph is always bounded by $\alpha_{\max}$, maintains a maximal matching in amortised time $\BigOh(\min_{\beta>1}\{\alpha_{\max} \cdot \beta + \log_{\beta} n\})$, which can be improved to $\BigOh(\log n/\log\log n)$ if the arboricity is always upper bounded by a constant. Under the same assumptions, He, Tang, and Zeh \cite{HeTZ14} improved this to $\BigOh(\alpha_{\max} + \sqrt{\alpha_{\max}\log n})$ amortised update time.
Without requiring that the arboricity be bounded at all times, the work by Kopelowitz, Krauthgsamer, Porat, and Solomon~\cite{KopelowitzKPS13} implies a deterministic algorithm with $\BigOh(\alpha^2 + \log^2 n)$ worst case update time, where $\alpha$ is the arboricity of the graph when receiving an edge-update.

\subsection{Dynamic \texorpdfstring{$\Delta+1$}{∆+1} colouring}

\begin{lemma} \label{thm:folkloreish}
Suppose one can maintain an edge-orientation of a dynamic graph, 
that has $t_u$ update time, 
that for each update performs at most $r_u$ edge re-orientations (direction changes), and that maintains a maximal out-degree of $\le n_o$. Then there is a dynamic $\Delta+1$-colouring algorithm whose update time is $\BigOh(t_u+r_u+n_o)$.
\end{lemma}

\begin{proof}
For a vertex $v$, say a colour is \emph{in-free} if no in-neighbor of $v$ has that colour. 
For a vertex of degree $d$, keep a doubly linked list of in-free colours from the palette 
$0,1,\ldots,d$. Keep an array \texttt{taken} of size $d+1$ where the $i$'th entry points to a doubly-linked list of in-neighbors of colour $i$, and an array \texttt{free} of size $d+1$ where the $i$'th entry points to the $i$'th colour in the list of in-free colours if the $i$'th colour is in-free.

The colour of a vertex $v$ is found by finding a colour that is both in-free and out-free: examine the $\le n_o$ out-neighbors, and use the \texttt{free}-array to temporarily move the $\le n_o$ out-taken colours to a list \texttt{out-taken}. Give $v$ an arbitrary free colour from the remaining list, and undo the \texttt{out-taken} list. This takes $\BigOh(n_o)$ time, and gives $v$ a colour between $0$ and its degree.

When an edge changes direction, this incurs $\BigOh(1)$ changes to linked lists and pointers. When an edge update incurs $r_u$ edge re-orientations, we thus have $\BigOh(r_u)$ such changes. When an edge is inserted/deleted from a properly coloured graph, at most one vertex needs to be recoloured, either because there is a colour conflict, or because its colour number is larger than its degree. This vertex can be recoloured in $\BigOh(n_o)$ time. Thus, the total time per edge insertion or deletion is $\BigOh(t_u + r_u + n_o)$.
\end{proof}

\begin{corollary}\label{cor:col}
    There is a deterministic dynamic $\Delta+1$ colouring algorithm with worst-case $\BigOh(\alpha + \log ^2 n \log \alpha)$ update time, where $\alpha$ is the current arboricity of the dynamic graph. %The algorithm also implies a $2$-approximate vertex cover in the same update time. 
\end{corollary}

\subparagraph{Related Work}
Previous work presented randomised algorithms with constant amortised update time per edge insertion/deletion \cite{BhattacharyaGKL22,HenzingerP22}. For deterministic algorithms, \cite{BhattacharyaCHN18} showed that if one is willing to use $(1+o(1))\cdot \Delta$, colours, a $\polylog(\Delta)$ amortised update time is possible. Solomon and Wein \cite{SolomonW20} extended the algorithm by \cite{BhattacharyaCHN18} and further showed that it is possible to maintain an $\alpha\log^2 n$ colouring in $\polylog(n)$ amortised update time.

%\cite{BhattacharyaGKL22,HenzingerP22,BhattacharyaCHN18,SolomonW20}

\subsection{Dynamic matrix vector product}

Suppose we have an $n \times n$ dynamic matrix $A$, and a dynamic $n$-vector $x$, and we want to maintain a data structure that allows us to efficiently query entries of $Ax$.  
The problem is related to the Online Boolean Matrix-Vector Multiplication (OMV), which is commonly used to obtain conditional lower bounds \cite{CliffordGL15,HenzingerKNS15,LarsenW17,Patrascu10}.
If $A$ is symmetric and sparse, in the sense that the undirected graph $G$ with $A$ as adjacency matrix has low arboricity, then we can use an algorithm for bounded out-degree orientation as a black-box to give an efficient data structure as follows:

%Remark: This statement is not in the published manuscript, but appears in the arxiv version.
\begin{lemma}[Implicit in Thm. A.3 in~\cite{KopelowitzKPS13}]
    Suppose one can maintain an edge-orientation of a dynamic graph with adjacency matrix $A$, that has $t_u$ update time, that for each update performs at most $r_u$ edge re-orientations (direction changes), and that maintains a maximal out-degree of $\le n_o$. Then there is a dynamic matrix-vector product algorithm that supports entry-pair changes to $A$ in $\BigOh(t_u+r_u)$ time, entry changes to the vector $x$ in $\BigOh(n_o)$ time, and queries to the an entry of product $Ax$ in $\BigOh(n_o)$ time.
\end{lemma}
\begin{proof}
    Let each node $i$ store the sum $s_i=\sum_{j\in N^-(i)}A_{ij}x_j$, i.e. the sum of the terms of $(Ax)_i=\sum_{j\in N(i)}A_{ij}x_j$ corresponding to incoming edges at $i$.  Changing entry $A_{ij}=A_{ji}$ in the matrix to or from $0$ corresponds to deleting or inserting an edge, which takes $t_u$ time and does at most $r_u$ edge re-orientations. Updating the $\BigOh(1)$ affected sums after inserting, deleting, re-orienting, or re-weighting an edge takes worst case $\BigOh(1)$ time. Any entry update to the matrix $A$ thus takes $\BigOh(t_u+r_u)$ time. 
    When a vector entry $x_j$ changes, we need to update the at most $n_o$ sums $\{s_i\}_{i\in N^+(j)}$, which can be done in worst case $\BigOh(n_o)$ time.
    Finally, the query for $(Ax)_i$ is computed as $(Ax)_i=s_i+\sum_{j\in N^+(i)}A_{ij}x_j$ in worst case $\BigOh(n_o)$ time.
\end{proof}

This result is used in~\cite[Theorem A.3]{KopelowitzKPS13} to give an algorithm for dynamic matrix vector product with running time $\BigOh(\alpha^2+\log^2 n)$ for updating the matrix, and $\BigOh(\alpha+\log n)$ for updating the vector and for queries.

Combining this theorem with our Theorem~\ref{thm:rank1} gives us an algorithm for dynamic matrix vector product with slightly improved time for updating the matrix:
\begin{corollary}\label{cor:dynamicmatrix}
    Let $A$ be a symmetric $n\times n$ matrix, and let $G$ be the undirected graph whose adjacency matrix is $A$. Let $x$ be an $n$ dimensional vector. Then we can support changes to $A$ in $\BigOh(\log^2 n\log\alpha)$ worst case time, changes to $x$ in $\BigOh(\alpha+\log n)$ worst case time, and for each $i\in \{1,\ldots,n\}$ we can report $\sum_{j=1}^nA_{ij}x_j$ in worst case $\BigOh(\alpha+\log n)$ time.
\end{corollary}

If we instead combine with our Theorem~\ref{thm:rank0} we get an algorithm for dynamic matrix vector product with slightly worse time for updating the matrix, but improved time for updating the vector and for queries:
\begin{corollary}
    Let $A$ be a symmetric $n \times n$ matrix, and let $G$ be the undirected graph whose adjacency matrix is $A$. Let $x$ be an $n$ dimensional vector. Then we can support changes to $A$ in $\BigOh(\log^3 n\log\alpha)$ worst case time, changes to $x$ in $\BigOh(\alpha)$ worst case time, and for each $i\in \{1,\ldots,n\}$ we can report $\sum_{j=1}^nA_{ij}x_j$ in worst case $\BigOh(\alpha)$ time.
\end{corollary}
%For constant $\alpha$ this is strictly better than~\cite{kopelowitz2014orienting}.

\subsection{Dynamic arboricity decomposition}

\begin{lemma}[\cite{henzinger2020explicit,christiansenICALP}] \label{thm:arbdecomp}
Suppose one can maintain an edge-orientation of a dynamic graph, 
that has $t_u$ update time, 
%that for each update performs at most $r_u$ edge re-orientations (direction changes),
and that maintains a maximal out-degree of $\le n_o$. Then there is an algorithm for maintaining a decomposition into $2n_o$ forests whose update time is $\BigOh(t_u)$.
\end{lemma}

\begin{proof}
    Firstly, as noted in \cite{henzinger2020explicit,christiansenICALP}: By assigning the $i$'th out-edge of a vertex $u$ to subgraph $S_i$, one obtains a decomposition into $n_o$ subgraphs, each of which is a pseudoforest. Every vertex has at most one out-edge in each pseudoforest $S_i$, and thus, the at most one cycle in each tree of the pseudoforest is a directed cycle according to the orientation.

    For maintaining this dynamic pseudoforest decomposition, there is only an $\BigOh(1)$ overhead per edge-reorientation, yielding an $\BigOh(n_o)$-time algorithm for maintaining $n_o$ pseudoforests.

    Then, as noted in \cite{henzinger2020explicit}, we may split each pseudoforest $S_i$ into two forests $f_i$ and $f_i'$ by the following simple algorithm: given a new edge $e$ in $f_i$, notice that there is at most one edge $e'$ in $f_i$ incident to its head. Now, one can safely insert $e$ in any of the two forests $\{f_i , f_i'\}$ that does not contain this at most one edge $e'$. Thus, consequently, neither $f_i$ nor $f_i'$ will contain a cycle.
\end{proof}

Thus, by applying Theorem~\ref{thm:rank0}, we obtain the following:

\begin{corollary}\label{cor:decomp}
    There is a deterministic algorithm for maintaining an arboricity decomposition into $\BigOh(\alpha)$ forests, whose worst-case update time is $\BigOh(\log ^3 n \log \alpha)$,
    where $\alpha$ is the current arboricity of the dynamic graph.
\end{corollary}

\subparagraph{Related Work}
While an arboricity decomposition of a graph; a partition of its edges into as few forests as possible; is conceptually easy to understand, computing an arboricity decomposition is surprisingly nontrivial. Even computing it exactly has received much attention~\cite{GabowW,Gabow95,Edmonds1965MinimumPO,PicardQueyranne82}. 
The state-of-the-art for computing an exact arboricity decomposition runs in $\tilde{O}(m^{3/2})$ time \cite{GabowW,Gabow95}.
In terms of not-exact algorithms there is a 2-approximation algorithm~\cite{ArikatiMZ97,Eppstein94} as well as an algorithm for computing an $\alpha+2$ arboricity decomposition in near-linear time~\cite{blumenstock2019constructive}.

For dynamic arboricity decomposition, Bannerjee et al.~\cite{Banerjee} give a dynamic algorithm for maintaining the current arboricity. The algorithm has a near-linear update time. They also provide a lower bound of $\Omega(\log{n})$. % for dynamically maintaining arboricity.
Henzinger Neumann Wiese~\cite{henzinger2020explicit} provide an $\BigOh(\alpha)$ arboricity decomposition in $\BigOh(\poly(\log n , \alpha))$ time; their result also goes via out-orientation, and they provide a dynamic algorithm for maintaining a $2\alpha'$ arboricity decomposition, given access to any black box dynamic $\alpha'$ out-degree orientation algorithm. 
Most recently, there are algorithms for maintaining $(\alpha+2)$ forests in $\BigOh(\operatorname{poly} (\log(n) , \alpha))$ update-time~\cite{christiansenICALP}, and $(\alpha + 1)$ forests in $\tilde{O}(n^{3/4}\operatorname{poly}(\alpha))$ time~\cite{christiansenMFCS}.

%Firstly, we easily get a pseudy-aboricity decomposition of $\BigOh(\alpha)$ pseudo-forests, following the lemma in CR22. Following this construction, every vertex will in every pseudo-tree have at most one out-neighbor according to the out-orientation.

%Then, it sounds intuitive that we can split each pseudo-tree into a tree and at most one extra edge, and thus, if we have $F$ pseudoforests, we obtain at most $2F$ forests.

%However, in order to split pseudo-trees across each cycle, we must divide the edges of $f_i\in F$ across the two forests $f_i^0$ and $f_i^1$ in an acyclic manner.

%To do this efficiently, one can utilise the technique described in HNW20, where they given a new edge $e$ in $f_i$ notice that there is at most one edge $e'$ in $f_i$ incident to its head. Now, one can safely insert $e$ in any of the two forests $\{f_i^1 , f_i^0\}$ that does not contain this at most one edge $e'$. Thus, consequently, neither $f_i^0$ nor $f_i^1$ will contain a cycle.

%(This follows by combining our result with CR'22)

\newpage

\label{place:refs}
\bibliographystyle{plain} %Icalp call.
\bibliography{refs}

\end{document}